%% file: main.tex
\documentclass[journal]{IEEEtran}
\usepackage{cite, bm,amsmath,graphicx,amssymb,mathtools,amsthm,subcaption,footmisc,framed,booktabs,dsfont,array}
\usepackage{mathrsfs}
\usepackage[margin=0.7in]{geometry}
\usepackage{fancyhdr,bm}
\usepackage{color}
\usepackage{threeparttable}
\usepackage{algorithm}
\usepackage{algpseudocode}
\algnewcommand\Input{\item[\textbf{Input:}]}
\usepackage[english]{babel}
    \usepackage{url}
\begin{document}
\include{macro}

\title{Codebook Design for Limited Feedback in Near-Field XL-MIMO Systems}

\author{
    \IEEEauthorblockN{Liujia Yao,~\IEEEmembership{Student Member,~IEEE}, Changsheng You,~\IEEEmembership{Member,~IEEE}, Zixuan Huang,~\IEEEmembership{Member,~IEEE},\\ Chao Zhou,~\IEEEmembership{Gruduate Student Member,~IEEE}, Zhaohui Yang,~\IEEEmembership{Member,~IEEE}, and Xiaoyang Li,~\IEEEmembership{Member,~IEEE}
    \vspace{-22pt}
    \thanks{Part of this work will be presented at the IEEE Global Communication Conference (GLOBECOM) 2025, Taipei, China, December 8-12, 2025~\cite{yao2024codebook}. L. Yao, C. You, C. Zhou, and X. Li are with the Department of Electrical and Electronic Engineering, Southern University of Science and Technology, Shenzhen, China (e-mail:\{yaolj2024,zhouchao2024\}@mail.sustech.edu.cn, \{youcs,lixy\}@sustech.edu.cn). \textit{(Corresponding authors: Changsheng You, Zixuan Huang).}}
    \thanks{Z. Huang is with the School of Electronics and Communication Engineering, Guangzhou University, China (e-mail:zxhuang@gzhu.edu.cn). }
    \thanks{Z. Yang is with the College of Information Science and Electronic Engineering, Zhejiang University, and also with the Zhejiang Provincial Key Laboratory of Info. Proc., Commun. \& Netw. (IPCAN), Hangzhou, 310027, China (e-mail: yang\_zhaohui@zju.edu.cn).}
    }
}
\maketitle

\begin{abstract}
   In this paper, we study efficient codebook design for limited feedback in extremely large-scale multiple-input-multiple-output (XL-MIMO) frequency division duplexing (FDD) systems.
   It is worth noting that existing codebook designs for XL-MIMO, such as polar-domain codebook, have not well taken into account user (location) distribution in practice, thereby incurring excessive feedback overhead.
   To address this issue, we propose in this paper a novel and efficient feedback codebook tailored to user distribution.
   To this end, we first consider a typical scenario where users are uniformly distributed within a specific polar-region, based on which a sum-rate maximization problem is formulated to jointly optimize angle-range samples and bit allocation among angle/range feedback.
   This problem is challenging to solve due to the lack of a closed-form expression for the received power in terms of angle and range samples.
   By leveraging a Voronoi partitioning approach, we show that uniform angle sampling is optimal for received power maximization.
   For more challenging range sampling design, we obtain a tight lower-bound on the received power and show that \emph{geometric} sampling, where the ratio between adjacent samples is constant, can maximize the lower bound and thus serves as a high-quality suboptimal solution.
   We then extend the proposed framework to accommodate more general non-uniform user distribution via an alternating sampling method.
   Furthermore, theoretical analysis reveals that as the array size increases, the optimal allocation of feedback bits increasingly favors range samples at the expense of angle samples.
Finally, numerical results validate the superior rate performance and robustness of the proposed codebook design under various system setups, achieving significant gains over benchmark schemes, including the widely used polar-domain codebook.
\end{abstract}

\begin{IEEEkeywords}
    Near-field communications, limited feedback, XL-MIMO, codebook design.
\end{IEEEkeywords}
\vspace{-1.2em}
\section{Introduction}
Extremely large-scale MIMO (XL-MIMO) is a promising technology for future wireless communication systems owing to its ability to provide high spectral and energy efficiency~\cite{wang2024tutorial,zengytutorial,you2024generationadvancedtransceivertechnologies}. 
In particular, via increasing the number of antennas at the base station (BS) by another order-of-magnitude, XL-MIMO is expected to shift conventional far-field electromagnetic propagation modeling to the new \emph{near-field} model with spherical wavefronts~\cite{musttwostagehieeam,lywHFB,zengytutorial}.
Such near-field spherical wavefronts bring both opportunities and challenges to wireless system designs. {Specifically, unlike conventional far-field maximum ratio transmission (MRT) beamforming that steers a directional beam~\cite{lywHFB,mustnfbeamtrainingdftcodebook}, the MRT-based near-field beam makes it possible to focus beam energy around/at a certain location/region, which is termed the \emph{beam-focusing} in the literature~\cite{zhangNearFieldSparseChannel2023,zhang2022beamfocusing}.} This beam-focusing effect brings a new degree-of-freedom (DoF) to control the beam energy distribution in both the angle and range domains \cite{mustnfisac}. 

Among others, accurate channel state information (CSI) is crucial for reaping the performance gains of XL-MIMO systems, which, however, is practically challenging to acquire due to the massive number of antennas and the complex near-field channel model. 
For high-frequency band systems, e.g., millimeter-wave (mmWave) and terahertz (THz), compressed sensing (CS) techniques have been widely used for far-field channel estimation by utilizing a predefined codebook~\cite{zhangNearFieldSparseChannel2023}, thanks to the angular sparsity of planar wavefronts in far-field channels. However, these works cannot be directly applied to near-field channel estimation due to the spherical wavefront. 
To address this issue, the authors in~\cite{cuiChannelEstimationExtremely2022} exploited the sparsity of near-field channel in angle and range domains, where a novel polar-domain codebook was designed for near-field channel estimation. As such, the general multi-path narrowband near-field channel can be estimated by using CS techniques, e.g., orthogonal matching pursuit (OMP)~\cite{lee2016channelomp}. However, practical channel estimation usually involves hybrid-field scenarios, where both far-field and near-field scatterers may coexist. 
This issue was tackled in~\cite{wei2022channel}, where the BS first estimates far-field channels using the conventional far-field discrete Fourier transform (DFT) codebook and then estimates near-field channels using the polar-domain codebook. 

However, the above-mentioned works mostly focus on time-division duplex (TDD) systems, where the downlink and uplink channels are reciprocal. 
For frequency-division duplex (FDD) systems, channel reciprocity does not hold in general, which advocates for efficient designs of CSI feedback schemes.
Specifically, due to the limited feedback overhead, users can only feed back a quantized (or reduced) version of CSI, e.g., the codeword index in a predefined codebook. 
This is known as the codebook-based limited feedback scheme, which has been widely adopted in practical wireless systems (e.g., 4G LTE and 5G NR)~\cite{limitedfeedbackSurvey}.
Specifically, a two-level codebook was proposed in \cite{alkhateebLimitedFeedbackHybrid2015} to reduce the feedback overhead for hybrid beamforming design, where the angle-of-arrival (AoA)/angle-of-departure (AoD) and the channel gains are reported by using a DFT codebook and a random vector quantization (RVQ) codebook, respectively.
For example, for conventional far-field FDD systems, various methods have been proposed for CSI feedback (e.g., \cite{alkhateebLimitedFeedbackHybrid2015,shen0218channelfeedback,kim2023feedback}) to quantize the complex-valued channel gains and AoA/AoD. 
This method was further extended in \cite{shen0218channelfeedback} to the multi-path channels.
More recently, a dominating path selection-and-feedback method was proposed in \cite{kim2023feedback}, where only the path gains of the dominating paths were selected for feedback to further reduce the feedback overhead by exploiting the uplink-downlink angular reciprocity.

For FDD XL-MIMO systems, recent efforts have been devoted to studying CSI feedback in the uplink. 
For instance, a deep-learning (DL)-based auto-encoder network was proposed in~\cite{deeplearningXLfeedback} to compress the high-dimensional near-field channel matrix into a compact latent representation, which is then fed back to reduce the feedback overhead.
However, this approach suffers from limited generality, as the network must be retrained for different system configurations.
To overcome this limitation, parametric-based near-field feedback has been investigated in the literature (e.g.,~\cite{polardomainXLfeedback,cuiChannelEstimationExtremely2022}), wherein key parameters of near-field channels, such as angles, ranges, and channel gains, are fed back.
As such, the near-field channel can be reconstructed at the BS based on the reported channel parameters. 
Among the existing approaches, a straightforward method proposed in \cite{polardomainXLfeedback} quantized the near-field geometric parameters (i.e., angles and ranges) by employing the polar-domain angle–range sampling technique introduced in \cite{cuiChannelEstimationExtremely2022}. 
To reduce the feedback overhead, a range-free method was proposed in \cite{twoAngleXLfeedback}, where only the angle information is fed back, while the range information is estimated at the BS based on angle information.
However, the above methods incur substantial feedback overhead to achieve satisfactory performance for practical implementation, as they overlook user distributions in practice, where users are typically distributed within a confined region (e.g., a sector with a certain range) rather than across the entire domain~\cite{zhangNearFieldSparseChannel2023,you2024generationadvancedtransceivertechnologies,mustmixednfffieldinterference}. 
By exploiting useful information of user distribution, more efficient feedback codebooks can be constructed to reduce the feedback overhead, which, however, have not been well studied. 

Motivated by the above, we study efficient codebook design for limited feedback in near-field FDD multi-user systems where the users are distributed in a given region.
Specifically, we first consider a typical case where users are uniformly and randomly distributed in a specific region of the polar space, and then extend our proposed method to more general non-uniform user distributions.
The main contributions of this paper are summarized as follows.

\begin{itemize}
    \item { First, for the case of uniformly distributed users, we formulate an optimization problem to maximize the sum-rate by jointly optimizing the angle–range sampling and the allocation of feedback bits. This problem is challenging to solve due to the lack of a closed-form expression for the received power in terms of the angle-range sampling. 
    To tackle this issue, we adopt a Voronoi partitioning approach and show that uniform angle sampling is optimal for power maximization. 
    For the more challenging range sampling, we obtain a tight lower bound on the received power, and show that \emph{geometric} range sampling, where the ratio between adjacent range samples remains constant, can maximize the lower bound and thus serves as a high-quality suboptimal solution. 
    The above results are further generalized to the \emph{non-uniform} user distribution cases, by adaptively refining the range/angle samples according to the underlying user distribution. 
    Moreover, an RVQ codebook is employed for efficient feedback of the effective channels for digital beamforming design to deal with inter-user interference (IUI). }
    
    \item Second, we provide theoretical analysis on the feedback bit allocation and show that the scaling order of the number of required feedback bits for analog beamforming is $\mathcal{O}(\log M)$, while that for digital beamforming is $\mathcal{O}(K)$, where $M$ and $K$ respectively denote the number of antennas and users. 
    Furthermore, we show that for arrays with a relatively small number of antennas, more feedback bits should be allocated to angle samples than to range samples. However, as the array size scales up, the optimal allocation shifts toward increasing the number of bits assigned to range samples, while reducing those assigned to angle samples.

	\item Third, extensive numerical results demonstrate that the proposed geometric codebook achieves superior sum-rate and beamforming gain over various benchmark schemes. In particular, the proposed codebook achieves near-optimal rate performance, especially in cases with limited feedback budgets or a large number of antennas. Moreover, the rate loss between the low-complexity closed-form codebook (derived under uniform-user assumption) and the proposed feedback codebook remains small even under highly non-uniform user distributions (e.g., hot-spot cases). 
    Furthermore, the proposed codebook is shown to be robust in scenarios with different numbers of antennas and various farthest user ranges. 
\end{itemize}

The remainder of this paper is organized as follows. In Section~\ref{sec:system_model}, we introduce the FDD near-field XL-MIMO system model. 
In Section~\ref{sec:proposed_algorithm}, we consider a three-phase transmission protocol for near-field XL-MIMO systems and formulate a feedback codebook design problem to maximize the expected achievable sum-rate for data transmission. 
In Section~\ref{sec:feedback_codebook}, we propose a novel feedback codebook design for received power maximization and interference cancellation. 
In Section~\ref{sec:analysis}, we provide extensive theoretical analysis to shed useful insights into the required feedback bits. 
Numerical simulations are presented in Section~\ref{sec:number_sim}. The main notations used in this paper are summarized in Table~\ref{tab:notation}.

\vspace{-1em}
\section{System Model and Problem Formulation}
\label{sec:system_model}
\subsection{System Model}
We consider a near-field FDD communication system as shown in Fig.~\ref{fig:system_set}, where a BS equipped with an XL-array serves $K$ single-antenna users. Specifically, the XL-array has $M$  antennas (assumed to be an odd number for ease of exposition), denoted by $\mathcal{M}\triangleq \{0,1,\dots,M-1\}$.
{It is  assumed that all users are located in the radiating near-field region of the BS, for which the BS-user distance $r_k$ is smaller than the so-called Rayleigh distance, i.e., $r_k < \frac{2D^2}{\lambda}$~\cite{zengytutorial}, where $D$ denotes the physical size (aperture) of the array and $\lambda$ denotes the signal wavelength.
Moreover, for ease of analysis, we assume that the user-BS distance $r_k\ge 1.2D$, $\forall k\in \mathcal{K}\triangleq \{1,...,K\}$, for which the user-BS channel can be modeled based on the uniform spherical wavefronts with the same channel amplitude across all antennas~\cite{emil12d}.}
Let $\theta_k=\sin\varphi_k $ denote the spatial angle of user $k$ w.r.t. the BS, where $\varphi_k$ denotes the physical angle.
The locations of users (i.e., $(\theta_k, r_k), k\in \mathcal{K}$) are assumed to be distributed in a specific region as follows. 
\begin{assumption}[Location distribution of users]
    \label{a:location_distribution}
    The users are randomly and uniformly distributed in a given region\footnote{For the purpose of the exposition, we consider the case where the users are uniformly located around the boresight, for which the spatial angle approximately follows the uniform distribution~\cite{2023sparsearray}. 
    The proposed codebook design method can be extended to the case of non-uniform user distribution, with details given in \textbf{Remark \ref{r:non_unif}} of Section III-B-2. 
    For an irregular user-distribution geometry shape, one can find a large enough polar-region to cover the user distribution, and this scenario falls into the non-uniform user-distribution case.}
    $$\mathcal{P}\triangleq\{(\theta, r) | \theta \sim \mathcal{U}(\theta_{\min},\theta_{\max}), r\sim \mathcal{U}(r_{\min},r_{\max})\},$$
    where $\theta_{\min}$ and $\theta_{\max}$ denote the minimum and maximum spatial angles, respectively, and $r_{\min}$ and $r_{\max}$ denote the minimum and maximum ranges, respectively. 
\end{assumption}

\begin{table}[t]
    \centering
    \caption{Main Notations}
    \begin{tabular}{c>{\centering\arraybackslash}p{0.7\linewidth}}
        \toprule
        \textbf{Symbol} & \textbf{Definition} \\
        \midrule
        $r_k$ & Range of user $k$ \\
        $\theta_k$ & Spatial angle of user $k$ \\
        $B_1$ (or $B_2$) & Phase 1 (or Phase 2) feedback bit overhead \\
        $p$ (or $q$) & Angle (or range) feedback bit overhead \\
        $\hat{\theta}_i$ (or $\hat{r}_i$) & The $i$-th angle (or range) sample \\
        $\mathcal{S}_{\rm a}$ (or $\mathcal{S}_{\rm r}$) & Angle (or range) sampling set \\
        $\varepsilon_\theta$ (or $\varepsilon_{r}$) & Angle (or range) sampling error \\
        $f$ & Analog beamforming gain \\
        $\Gamma$ &  Expected beamforming gain w.r.t. user locations \\
        $\tilde{\varepsilon}_{\theta}$ (or $\tilde{\varepsilon}_{r}$) & Minimum angle (or range) sampling error \\
        $\mathcal{Q}$ & Angle domain set $[\theta_{\min},\theta_{\max}]$ \\
        $\mathcal{C}_i$ & The $i$-th angle cell \\
        $\mathcal{R}$ & Range domain set $[r_{\min},r_{\max}]$ \\
        $\mathcal{I}_i$ & The $i$-th range cell \\
        $\mathring{\varepsilon}_r$ & Surrogate range sampling error \\
        \bottomrule
    \end{tabular}
    \label{tab:notation}
\end{table}

\textbf{\underline{Channel model:}} Let $\ch_k\in \mathbb{C}^{M\times 1}$ denote the multi-path channel from the BS to user $k$. Based on the near-field channel model in high-frequency bands~\cite{zhangNearFieldSparseChannel2023},  $\ch_k, \forall k \in \mathcal{K}$, can be modeled as
\begin{equation}
    \label{eq:channel}
    \ch_k = \sum_{\ell=1}^{L} \sqrt{M} \beta_{k,\ell} \bfa(\theta_{k,\ell}, r_{k,\ell}),
\end{equation}
where $\ell = 1$ and $\ell > 1$ refer to the line-of-sight (LoS) and non-LoS (NLoS) path-indices, respectively, and $\beta_{k,\ell}$ is the complex-valued channel gain of user $k$ in path $\ell$. Herein, $\bfa(\theta_{k,\ell}, r_{k,\ell})\in \mathbb{C}^{M\times 1}$ denotes the near-field channel steering vector, which is modeled as follows based on the spherical wavefront
\begin{equation}{
    \label{eq:steering_vector}
    \begin{aligned}
    \bfa(\theta_{k,\ell}, r_{k,\ell}) = \frac{1}{\!\!\sqrt{M}}
    \bigg[e^{\jmath 2 \pi\frac{r_{k,\ell}^{(0)}-r_{k,\ell}^{}}{\lambda}},\dots,
    e^{\jmath 2 \pi\frac{r_{k,\ell}^{(M-1)}-r_{k,\ell}}{\lambda}}\bigg]^H\!\!,
    \end{aligned}}
\end{equation}
with $r_{k,\ell}^{(m)}$ representing the range between the $\ell$-th scatterer (or the user itself) of user $k$ and the $m$-th antenna of the BS. Moreover, $r_{k,\ell}$ and $\theta_{k,\ell}$ are the reference (w.r.t. the origin point, see Fig.~\ref{fig:system_set}) range and spatial angle, respectively. Specifically, let $\delta_m = \frac{2m-M-1}{2} $ denote the coordinate of the $m$-th antenna along the $y$-axis. Based on the Fresnel approximation, $r_{k,\ell}^{(m)}= \sqrt{ r_{k,\ell}^2 + \delta_m^2 d_0 ^2 - 2 r_{k,\ell} \theta_{k,\ell} \delta_m d_0}$ in~\eqref{eq:steering_vector} can be approximated as~\cite{zhangNearFieldSparseChannel2023}
\begin{align}
    r_{k,\ell}^{(m)} \approx r_{k,\ell}^{(0)}- \delta_m d_0 \theta_{k,\ell} + \frac{\delta_m^2 d_0^2}{2r_{k,\ell}}(1-\theta_{k,\ell}^2), \forall k,\forall \ell,\forall m \label{eq:steering_vector_approx}
\end{align}
where $d_0=\frac{\lambda}{2}$ denotes the inter-antenna spacing.

\begin{figure}
    \centering
     \vspace{-13pt}
    \includegraphics[width=0.8\linewidth]{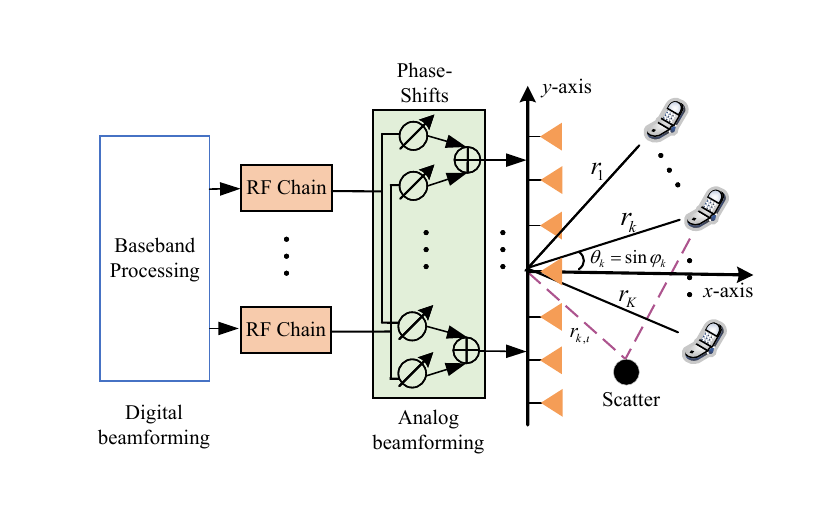}
    \vspace{-20pt}
    \caption{An XL-MIMO near-field communication system.}\label{fig:system_set}
    \vspace{-16pt}
\end{figure}
\begin{figure*}[th]
    \centering
    \vspace{-14pt}
    \begin{subfigure}[b]{0.75\linewidth}
        \centering
        \includegraphics[width=\linewidth]{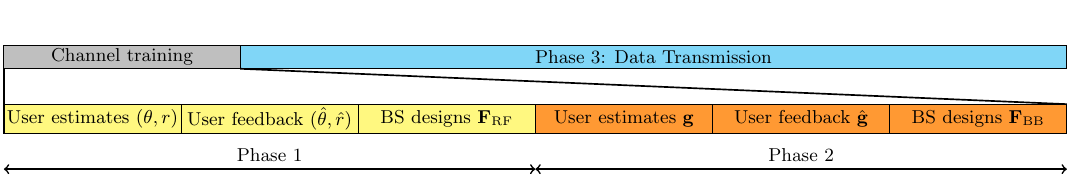}
        \caption{Considered transmission protocol for FDD near-field systems.}
        \label{fig:tx_protocal}
    \end{subfigure}
    \begin{subfigure}[b]{0.75\linewidth}
        \centering
        \includegraphics[width=\linewidth]{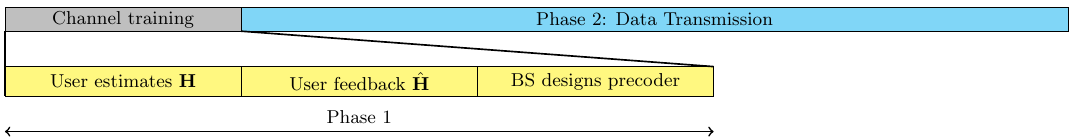}
        \caption{ A conventional FDD transmission protocol in \cite{DFTcodebookXLbeamtraining,polardomainXLfeedback,pengDeepLearningBasedCSI2024}.}
        \label{fig:tx_protocal_old}
    \end{subfigure}
    \caption{Comparison of FDD transmission protocols.}
    \label{fig:protocols}
    \vspace{-16pt}
\end{figure*}
\textbf{\underline{Signal model:}} For FDD XL-MIMO systems in high-frequency bands, the channel gains of NLoS paths are much smaller
(e.g., 20 dB weaker) than those of the LoS path~\cite{lywTut,zengytutorial}. As such, we first consider the LoS channel case in the following, while the results will be extended to the general multi-path case in Section~\ref{sec:feedback_codebook}. Under this setup, the channel for each user $k$ in~\eqref{eq:channel} reduces to $\ch_k =\sqrt{M} \beta_{k} \bfa(\theta_{k},r_{k})$, where we omit the path-index for simplicity. {We consider a cost-effective hybrid beamforming architecture for XL-arrays \cite{lywHFB,lywTut,zhou2024multibeamtrainingnearfieldcommunications}, which comprises $K$ radio frequency (RF) chains for serving $K$ single-antenna users.} Let $\mathbf{x} \in \mathbb{C}^{K \times 1}$ denote the transmitted signals to all the $K$ users with $\expect\left[\mathbf{x x}^H\right]=\frac{P_{\text {tol }}}{K} \mathbf{I}_K$, where $P_{\text {tol}}$ denotes the total transmit power of the BS. The received signal at user $k$ is given by
\begin{equation}
    \begin{aligned}
         & y_k = \mathbf{h}_k^H \mathbf{F}_{\rm RF} \mathbf{F}_{\rm BB} \mathbf{x} + n_k, \forall k \in \mathcal{K},
    \end{aligned}
\end{equation}
where $\frf \in \mathbb{C}^{M \times K}$ and $\mathbf{F}_{\mathrm{BB}} \in \mathbb{C}^{K \times K}$ denote the analog and digital beamforming matrices, respectively, and $n_k\sim \mathcal{CN}(0,\sigma^2)$ is the received additive white Gaussian noise (AWGN) at user $k$ with zero mean and variance $\sigma^2$. As such, the achievable rate of the $k$-th user in bits/second/Hertz (bps/Hz) is given by
\begin{equation}
    \begin{aligned}
        \label{eq:rate}
         & R_k\!=\!\log _2\left(1+\frac{\frac{P_{\text {tol }}}{K}\left|\mathbf{h}_k^H \frf \fbbk{k}\right|^2}{\frac{P_{\text {tol }}}{K} \sum_{i \neq k}\left|\mathbf{h}_k^H \frf \fbbk{i}\right|^2+\sigma^2}\right),
    \end{aligned}
\end{equation}
where $\fbbk{k}$ denotes the $k$-th column of $\mathbf{F}_{\rm BB}$.\vspace{-1em}

\subsection{Problem Formulation}
In FDD systems, the CSI has to be fed back from the users to the BS in the uplink, which may incur high feedback overhead and hence degrade system spectral efficiency~\cite{limitedfeedbackSurvey,guojiajia2020feedback}. To reduce feedback overhead, we consider a \textit{codebook-based} feedback scheme (whose details will be specified in Section~\ref{sec:proposed_algorithm}). 
Specifically, similar to~\cite{alkhateebLimitedFeedbackHybrid2015,shen0218channelfeedback}, the users send the information of their individual analog beamforming codewords and effective channel codewords back to the BS based on estimated CSI to maximize their sum-rate. All feedback codewords are chosen from pre-designed codebooks, in which the analog beamforming codebook is denoted by $\Bcal_1$ and the effective channel codebook is denoted by $\Bcal_2$. 

For the design of CSI feedback scheme, we aim to maximizing the expected sum-rate, for which the users' locations are distributed according to a specific distribution given in \textbf{Assumption~\ref{a:location_distribution}}.  
Specifically, given analog beamforming matrix $\frf$, the digital beamforming can be designed based on the effective channels of users (denoted by $\mathbf{g}_k\triangleq \mathbf{h}^H_k \frf,\forall k \in \mathcal{K}$). 
Define $\fbb(\mathbf{G})$ as the digital beamforming matrix given the effective channel matrix $\mathbf{G}\triangleq [\mathbf{g}_1,\mathbf{g}_2,\dots,\mathbf{g}_K]$. Based on the above, the optimization problem for multi-user sum-rate maximization can be mathematically formulated as
\begin{subequations}
    \begin{align}
        \label{eq:problem} \text{(P1):} & \max_{\Bcal_1,\mathcal{B}_2} \quad \expect_{(\theta,r)}\left[ \max_{\frf,\fbb(\mathbf{G})} \sum_{k\in \mathcal{K}} R_k\right]                          \\
                                    &\; \text{s.t.} \;\;\;~~~ \frfk{k}\in \Bcal_1, ~\forall k\in\mathcal{K}, \label{eq:RF_codebook} \\ 
                                    & \;\phantom{\text{s.t. }} \;\;~~~\mathbf{g}_k\in \Bcal_2, ~\forall k\in\mathcal{K}, \label{eq:BB_codebook} \\
                                    & \;\phantom{\text{s.t. }} ~~~~\left\|\mathbf{F}_{\mathrm{RF}} \mathbf{f}_{\mathrm{BB},  k}\right\|_F^2=1, ~\forall k \in \mathcal{K}, \label{eq:RF_norm_cons}
    \end{align}
\end{subequations}
where the expectation is taken over the random users' locations (see \textbf{Assumption~\ref{a:location_distribution}}), constraints \eqref{eq:RF_codebook} and \eqref{eq:BB_codebook} specify the codebook-based constraints, and constraint \eqref{eq:RF_norm_cons} specifies the unit-power constraint for each user accounting for equal power allocation\footnote{To simplify the analysis and focus on the codebook design, we assume equal power allocation for all users in this paper, whose optimization can also be incorporated into the proposed framework as an extension, via e.g., fractional programming~\cite{lywHFB}.}.

It is worth noting that the customized codebooks $\Bcal_1$ and $\Bcal_2$ are designed based on the \emph{user-location distribution} given in \textbf{Assumption~\ref{a:location_distribution}}, and $\frf$ and $\fbb$ are designed based on each \emph{possible realization} of locations of users. 
Problem (P1) is generally difficult to solve due to the lack of a closed-form expression for the expected sum-rate, the non-convex constraints~\eqref{eq:RF_codebook} and~\eqref{eq:BB_codebook}, as well as the intricate coupling between the analog beamforming matrix $\frf$ and the digital beamforming matrix $\fbb$. 
\section{CSI-feedback-based Hybrid Beamforming}
\label{sec:proposed_algorithm}
{ To address the intricate coupling between the analog and digital beamforming in Problem (P1) and avoid prohibitively high overhead for information exchange between the users and the BS, we consider a three-phase transmission protocol shown in Fig.~\ref{fig:protocols}(a), which is elaborated as follows.}

\textbullet \textbf{ Phase 1 (Location-based limited feedback for analog beamforming design)} In Phase 1, each user estimates its downlink CSI, based on which the codeword (index) from $\mathcal{B}_1$ that maximizes its received power is fed back to the BS for designing its analog beamforming $\frfk{k}$. Let $\hat{\theta}$ (or $\hat{r}$) be an arbitrary angle (or range) sample in the angle (or range) domain. Let $B_1=p+q$ denote the number of feedback bits in this phase, with $p$ and $q$ denoting the number of bits allocated for angle and range feedback, respectively. The feedback codebook in Phase 1 (or equivalently the codebook for analog beamforming) can be constructed as $\mathcal{\Bcal}_1\triangleq \{\mathbf{b}_{1}(i,j)\mid 1 \le i\le 2^p, 1 \le j\le 2^q\}$, where 
    $\mathbf{b}_{1}(i,j)=\mathbf{a}(\hat{\theta}_i,\hat{r}_j), \forall \hat{\theta}_i\in \mathcal{S}_{\rm a},\forall\hat{r}_j\in \Scal_{\rm r}$. Herein, $\Scal_{\rm a}$ and $\Scal_{\rm r}$ denote the sampling sets for the angles and ranges given the bit budgets $|\Scal_{\rm a}|=2^p, |\Scal_{\rm r}|=2^q$, respectively. Based on the above, the analog beamforming codebook design problem to maximize the expected received power (under limited feedback) for an arbitrary user $k$ (following a certain distribution) can be mathematically formulated as
    \begin{subequations}
        \begin{align}
            \label{eq:P1_RF_reformulated}
            (\text{P2}):& \max_{\Bcal_1} \quad \expect_{(\theta,r)} \left[ \max_{\mathbf{f}_{\mathrm{RF},k}\in \Bcal_1} \left|\mathbf{h}_k^H \mathbf{f}_{\mathrm{RF},k}\right|\right] \\
            &\text{ s.t.} \quad \;\; |\Bcal_1| = 2^{B_1} \label{eq:P1_RF_constraint2}.
        \end{align}
    \end{subequations}

    \textbullet \textbf{ Phase 2 (Effective-channel limited feedback for digital beamforming design)} In Phase 2, each user estimates its effective channel $\mathbf{g}_k$ and selects a codeword from the codebook $\mathcal{B}_2$ to feed it back to the BS.
    The BS then collects all users' codewords for Phase 2 to reconstruct the effective channel matrix, based on which the BS designs the digital beamforming matrix to suppress the residual IUI and maximize the expected sum-rate. The codebook in Phase 2 can be constructed as $\mathcal{B}_2 \triangleq \{\mathbf{b}_2(i) \mid 1 \le i \le 2^{B_2} \}$, where $B_2$ denotes the feedback bit overhead in Phase 2 and $\mathbf{b}_2(i)$ is the $i$-th codeword of $\Bcal_2$. The optimal codeword for user $k$ is chosen as $\hat{\mathbf{g}}_k = \arg\max_{\mathbf{b}_2 \in \mathcal{B}_2} |\mathbf{g}_k^H \mathbf{b}_2|^2$~\cite{alkhateebLimitedFeedbackHybrid2015}, and the effective channel matrix is reconstructed as $\hat{\mathbf{G}} \triangleq [\hat{\mathbf{g}}_1, \dots, \hat{\mathbf{g}}_K]$.
    Then the multi-user sum-rate maximization problem can be reformulated as 
    \begin{subequations}
        \begin{align}
            \label{eq:P1_BB_reformulated}
            (\text{P3}): & \max_{\Bcal_2} \quad \expect_{(\theta,r)}\left[ \max_{\fbb(\hat{\mathbf{G}})} \sum_{k=1}^{K} R_k \right] \\
            &\text{ s.t. } \;\quad \left\| \mathbf{f}_{\mathrm{BB}, k}\right\|_F^2=1, \forall k \in \mathcal{K}, \\
            &\phantom{ \text{s.t. }}\quad \;\;\; {\hat{\mathbf{g}}_k}\in \mathcal{B}_2, \forall k \in \mathcal{K}, \label{eq:P1_BB_reformulated2}\\
            &\phantom{ \text{s.t. }} \quad\;\;\;|\Bcal_2| = 2^{B_2}.
        \end{align}
    \end{subequations}

    \textbullet \textbf{ Phase 3 (Data transmission)}: In this phase, the BS sends data to all users based on the designed $\frf$ and $\fbb$. For ease of analysis and focusing on the feedback codebook design, we make the following assumptions. The imperfect CSI case will be discussed later in Section~\ref{sec:feedback_codebook}.
    
    { Different from the conventional FDD transmission protocol shown in Fig.~\ref{fig:protocols}(b), in our considered protocol, we first aim to maximize the received power at each user, and then suppress the residual IUI via digital beamforming design.}
\begin{assumption}
     Similar to \cite{alkhateebLimitedFeedbackHybrid2015,au-yeungPerformanceRandomVector2007a}, we make the following assumptions for CSI.
    \begin{enumerate}
        \item Each user $k$ perfectly estimates its downlink CSI $\ch_k$ and its effective channel $\mathbf{g}_k$ via existing efficient channel estimation schemes (see, e.g., \cite{cuiChannelEstimationExtremely2022,lu2023nearfieldestimation,zhangNearFieldSparseChannel2023}). 
        \item The BS perfectly knows the feedback information via a low-rate and reliable uplink channel~\cite{limitedfeedbackSurvey}.
    \end{enumerate}
\end{assumption}
\vspace{-1em}

\section{Proposed Codebook Design}
\label{sec:feedback_codebook}
In this section, a new polar-domain codebook for the limited feedback in Phase 1 is proposed for solving Problem (P2), while an RVQ codebook is designed for Phase 2 to solve Problem~(P3) by utilizing the ZF technique.
We show that the proposed \textit{geometric} codebook, which utilizes a novel \textit{geometric} range sampling set, achieves significantly higher expected received power at each user compared to the existing polar-domain codebook in \cite{cuiChannelEstimationExtremely2022}. 
\vspace{-1em}
\subsection{Codebook Design for Phase 1}
Each user $k$ selects the codeword $\mathbf{b}_{1,k}^\star$ from $\mathcal{B}_1$ that maximizes its received power. Ignoring the constant channel gain, this is equivalent to finding a codeword that maximizes the beamforming gain, i.e., $\mathbf{b}_{1,k}^\star = \arg\max_{\mathbf{b}_1 \in \mathcal{B}_1} |\bfa^H(\theta_k,r_k)\mathbf{b}_1|$. The index of this optimal codeword is then fed back to the BS.
{ The BS then uses this codeword to design $\frfk{k}=\mathbf{b}_{1,k}^\star$. Specifically, the design of $\mathcal{B}_1$ is equivalent to the design of angle and range sampling sets $\Scal_{\rm a}$ and $\Scal_{\rm r}$ (as well as their corresponding bit allocation $p,q$ given a  practical constraint on the total number of bits $p+q =B_1$). }
Thus, Problem (P2) can be equivalently reformulated as follows for each user, with the user index omitted for brevity
\begin{subequations}
    \begin{align}
        (\text{P4}): & \max_{\Scal_{\rm a},\Scal_{\rm r},p,q} \;\; \expect_{(\theta,r)} \left[\max_{\hat{\theta}\in \Scal_{\rm a},\hat{r}\in \Scal_{\rm r}} \left|\mathbf{a}^H\left(\theta, r\right) \mathbf{a}(\hat{\theta},\hat{r})\right|\right] \label{eq:P4}\\
        &\;\;\;\text{s.t.}\;\; \;\;\;\;\left|\Scal_{\rm a}\right| = 2^{p},\left
|\Scal_{\rm r}\right| = 2^{q}, \label{eq:P4_cardinality_cons_r}\\
&\phantom{\;\text{s.t.}\;\; \qquad} p+q = B_1,
        \label{eq:P4_cardinality_cons} \\
        &\phantom{\;\text{s.t.}\;\; \qquad} p,q \in \mathbb{Z}. \label{eq:P4_cardinality_cons_pq}
    \end{align}
\end{subequations}
Problem (P4) is a challenging non-convex optimization problem. The difficulties stem from the fact that the objective function lacks a closed-form expression, and the optimization of the sampling sets ($\mathcal{S}_{\rm a}$, $\mathcal{S}_{\rm r}$) is coupled with the bit allocation ($p,q$).
{To tackle these difficulties, we solve Problem (P4) by alternately optimizing ($\mathcal{S}_{\rm a}$, $\mathcal{S}_{\rm r}$) and $(p,q)$ based on a two-layer optimization framework, as elaborated below.}
\subsubsection{\textbf{Problem Reformulation}}
\label{sec:problem_reformulation}
The inner-layer subproblem optimizes the sampling sets $\mathcal{S}_{\rm a}$ and $\mathcal{S}_{\rm r}$ for a given feasible bit allocation $(p,q)$, and the outer-layer subproblem optimizes the bit allocation $(p,q)$ based on the solution to the inner-layer problem.

{\textbf{Optimizing sampling sets given bit allocation:} }Specifically, we first define $\epsT \triangleq | \theta-\hat{\theta} |$ and $\varepsilon_{r} \triangleq | \frac{1}{r} - \frac{1}{\hat{r}} |$ as the \emph{sampling errors} related to the angle and range parameters, respectively. Then, the analog beamforming gain in \eqref{eq:P4} under a feedback codebook $\mathcal{\Bcal}_1$ can be approximated as follows.
\begin{lemma}
    \label{l:beamforming_gain}
    Given the feedback codebook $\Bcal_1$ with a sufficient\footnote{For example, $p$ being 10--12 is sufficient~\cite{cuiChannelEstimationExtremely2022,shen0218channelfeedback} for this lemma, which is practically acceptable.} angle bit overhead $p$, the analog beamforming gain in~\eqref{eq:P4} can be approximated as
    \begin{align}
        &\left|  \bfa^H(\theta,r) \bfa(\hat{\theta},\hat{r})\right|\notag\\
        \approx & \frac{2}{M} \left| \sum_{m=0}^{(M-1)/2} e^{\jmath m \pi \epsT} e^{\jmath k_c m^2 d_0^2 \frac{1}{2} \vartheta \varepsilon_{r}}  \right| \triangleq f(\varepsilon_\theta,\vartheta\varepsilon_{r}),      \label{eq:approximation}
    \end{align}
    where $k_c \triangleq 2\pi/\lambda$ and $\vartheta \triangleq 1-\theta^2$.
\end{lemma}
\begin{proof}
        See Appendix~\ref{sec:appendix_bfgain}.
\end{proof}
Based on \textbf{Lemma~\ref{l:beamforming_gain}}, given $\Bcal_1$ with a sufficiently large angle bit overhead, the expected analog beamforming gain in~\eqref{eq:P4} under arbitrary sampling sets $\mathcal{S}_{\rm a},\mathcal{S}_{\rm r}$ can be approximated as
\begin{align}
   &\expect_{(\theta,r)}\left[\max_{\hat{\theta}\in \Scal_{\rm a},\hat{r}\in \Scal_{\rm r}}\left|\bfa^H(\theta,r) \bfa(\hat{\theta},\hat{r})\right|\right] \notag \\
    \approx &\expect_{(\theta,r)} \left[\max_{\hat{\theta}\in \Scal_{\rm a},\hat{r}\in \Scal_{\rm r}} f(\varepsilon_\theta,\vartheta\varepsilon_{r})\right] \triangleq  \Gamma(\Scal_{\rm a},\Scal_{\rm r}|p,q).  \label{eq:beamforming_gain}
\end{align}
Based on the approximation in~\eqref{eq:beamforming_gain} and given any feasible bit allocation $(p,q)$, Problem (P4) approximately reduces to the following inner-layer optimization problem
\begin{equation}
        \label{eq:P5}
        (\text{P5}): \max_{\Scal_{\rm a},\Scal_{\rm r}} \;\; \Gamma(\Scal_{\rm a},\Scal_{\rm r}|p,q) 
        \;\;\text{s.t.} \;\; \eqref{eq:P4_cardinality_cons_r}.  \notag
\end{equation}

\textbf{{Optimizing bit allocation with optimized sampling sets:}}
Let $\tilde{\Gamma}(p,q) = \Gamma(\Scal_{\rm a}^{*},\Scal_{\rm r}^* | p,q)$ denote the expected analog beamforming gain given the bit allocation $(p,q)$, with $\Scal_{\rm a}^*$ and $\Scal_{\rm r}^*$ obtained by solving Problem (P5). Based on the above, Problem (P4) reduces to the following outer-layer optimization problem
\begin{equation}
        \label{eq:P6}(\text{P6}):  \max_{p,q} \; \tilde{\Gamma}(p,q)\;\;\text { s.t. } \;\;\eqref{eq:P4_cardinality_cons},\eqref{eq:P4_cardinality_cons_pq}.\notag
\end{equation}
In the following, we first solve Problem (P5), given any feasible feedback bit allocation, and then solve Problem~(P6).

\subsubsection{\textbf{Proposed Solution to Problem (P5)}}
\label{sec:codebook_1_design}
Problem (P5) is difficult to solve since $\Gamma(\Scal_{\rm a},\Scal_{\rm r} | p,q)$ admits no closed-form expression. To tackle this issue and characterize useful properties of $f(\varepsilon_\theta,\vartheta\varepsilon_{r})$, we first plot the curves of $f(\varepsilon_\theta,\vartheta\varepsilon_{r})$ w.r.t. $\varepsilon_\theta$ and $\varepsilon_{r}$ in Fig.~\ref{fig:f_AoD_r}, respectively. For ease of analysis, we assume $\vartheta=1$ in the following discussions.
\begin{observation}
    \label{p:approximation_decreasing}
        When $\varepsilon_{\iota}$ with $\iota\in\{\theta,{r}\}$ is smaller than a given threshold $\varepsilon_{\iota}^{\rm th}$, the function $f(\varepsilon_\theta,\varepsilon_{r})$ approximately linearly decreases with $\varepsilon_\theta$ and $\varepsilon_{r}$.
    \end{observation}

\begin{figure}[t]
    \centering
    \begin{minipage}{0.47\linewidth}
        \centering
        \vspace{-10pt}
        \includegraphics[width=\linewidth]{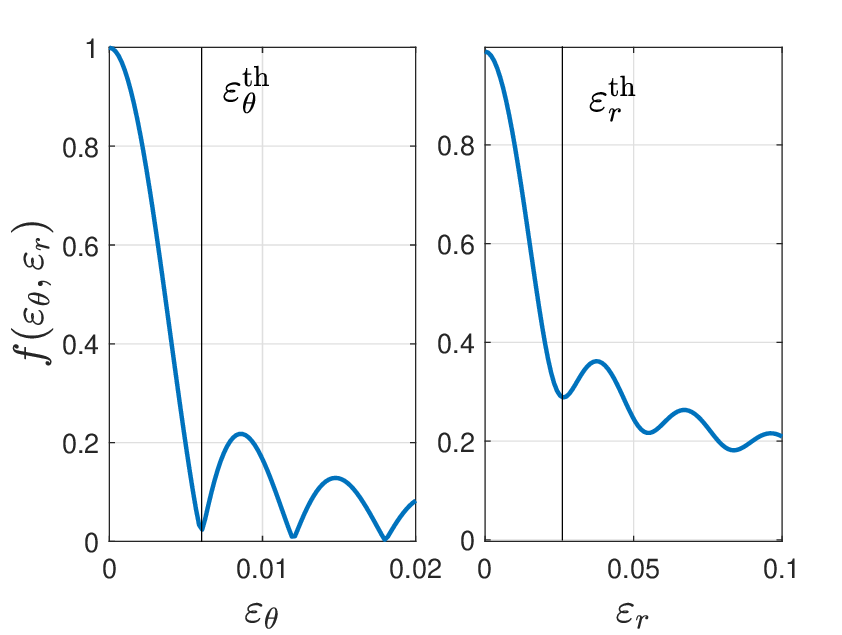}
        \caption{\small $f$ vs. $\varepsilon_{r}$ (given $\varepsilon_\theta=0$) and $\varepsilon_\theta$ (given $\varepsilon_{r} = 0$) with $\vartheta=1$ (other parameters refer to Section~\ref{sec:number_sim}).}
        \label{fig:f_AoD_r}
    \end{minipage}
    \hfill
    \begin{minipage}{0.45\linewidth}
        \centering
        \vspace{-10pt}
        \includegraphics[width=\linewidth]{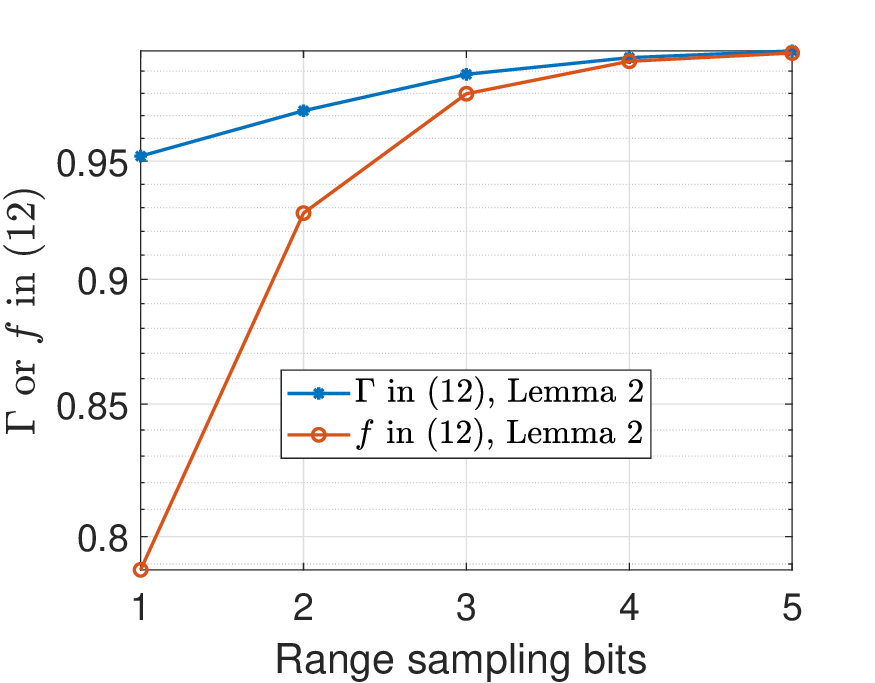}
        \caption{\small Accuracy verification of \textbf{Lemma~\ref{lemma:Gamma_approximation}} with $p=11$ (other parameters refer to Section~\ref{sec:number_sim}).}
        \label{fig:lemma1_veri}
    \end{minipage}
    \vspace{-2em}
\end{figure}
{ For example, for a typical XL-MIMO system specified in Fig.~\ref{fig:f_AoD_r}, by solving $\frac{\partial f(\epsT,0)}{\partial \epsT} = 0$ and $\frac{\partial f(0,\vartheta\varepsilon_{r})}{\partial \varepsilon_{r}} = 0$, we can numerically obtain $\epsT^{\rm th}\approx 0.0058$ and $\varepsilon_{r}^{\rm th}\approx 0.027$.
In addition, when $\epsT>\epsT^{\rm th}$ (or $\varepsilon_{r}>\varepsilon_{r}^{\rm th}$), the analog beamforming gain (i.e., $f(\varepsilon_\theta,\varepsilon_{r})$) is almost zero, thus incurring significant rate loss due to a large sampling error.
To simplify the analysis, we mainly consider the regime $\varepsilon_\iota \le \varepsilon^{\rm th}_\iota$ (i.e., the sampling error is acceptable) in the codebook design, which is also widely adopted in the literature~\cite{au-yeungPerformanceRandomVector2007a,alkhateebLimitedFeedbackHybrid2015,jindalMIMOBroadcastChannels2006a,sohrabiHybridDigitalAnalog2016}.}

Based on \textbf{Observation~\ref{p:approximation_decreasing}}, maximizing $f(\epsT,\vartheta\varepsilon_{r})$ is approximately equivalent to minimizing $\epsT$ and $\varepsilon_{r}$. Thus, we can define $\tilde{\varepsilon}_{\theta}(\mathcal{S}_{\rm a})\triangleq \min_{\hat{\theta}\in \Scal_{\rm a}} \epsT$ and $\tilde{\varepsilon}_{r}(\mathcal{S}_{r}) = \min_{\hat{r}\in \Scal_{\rm r}} | \frac{1}{r} - \frac{1}{\hat{r}} |$ as the minimum angle/range sampling errors in $\Scal_{\rm a}$ and $\Scal_{\rm r}$, respectively. Specifically, we define $\mathcal{Q}\triangleq [\theta_{\min},\theta_{\max}]$ and $\mathcal{R}\triangleq [r_{\min},r_{\max}]$ as the angle and range domains of interest, respectively. 
The following lemma reveals the relationship between the expected analog beamforming gain and the sampling errors.
\begin{lemma}
    \label{lemma:Gamma_approximation}
    When $p$ and $q$ are sufficiently large\footnote{An example is provided in Fig.~\ref{fig:lemma1_veri} at the top of the previous page, which suggests that $p=11$ and $q=2$ are sufficient.}, the expected analog beamforming gain $\Gamma$ in~\eqref{eq:beamforming_gain} can be approximated by
    \begin{equation}
        \label{eq:Gamma_approximation}
        \Gamma(\Scal_{\rm a},\Scal_{\rm r}|p,q) \approx f(\expect_{\theta}[\tilde{\varepsilon}_{\theta}], \check{\vartheta} \expect_{r}[\tilde{\varepsilon}_{r}]),
    \end{equation}
    where $\check{\vartheta}\triangleq \mathbb{E}_\theta [\vartheta]$, {$\expect_{\theta}[\tilde{\varepsilon}_{\theta}] = \int_{\mathcal{Q}}\tilde{\varepsilon}_{\theta} \mu(\theta) d\theta$ and $\expect_{r}[\tilde{\varepsilon}_{r}] = \int_{\mathcal{R}}\tilde{\varepsilon}_{r} \mu(r) dr$}, with $\mu(x)$ denoting the probability density function (p.d.f.) of $x$.
\end{lemma}
\begin{proof}
    See Appendix~\ref{sec:appendix_expect_decouple}.
\end{proof}
\vspace{-0.5em}
As $ f$ generally decreases with $\mathbb{E}_\theta[\tilde{\varepsilon}_{\theta}]$ and $\check{\vartheta} \mathbb{E}_r[\tilde{\varepsilon}_{r}]$ (with $\check{\vartheta}$ being a constant), based on \textbf{Lemma~\ref{lemma:Gamma_approximation}}, Problem (P5) can be equivalently decomposed into two subproblems (i.e., Problems (P7) and (P9)) for minimizing $\expect_\theta[\tilde{\varepsilon}_{\theta}(\mathcal{S}_{\rm a})]$ and $\expect_r[\tilde{\varepsilon}_{r}(\mathcal{S}_{\rm r})]$, respectively. These two subproblems are only related to $\mathcal{S}_{\rm a}$ and $\mathcal{S}_{\rm r}$, and thus can be solved separately, as elaborated below.

{\textbf{Proposed angle sampling method:}}
\label{sec:minimization_epsT}
The subproblem of minimizing $\expect_\theta[\tilde{\varepsilon}_{\theta}]$ can be mathematically formulated as
\begin{subequations}
    \begin{align}
        \text{(P7):}\;\; & \operatorname*{\min}_{\Scal_{\rm a}} \;\; \int_{\mathcal{Q}} \tilde{\varepsilon}_\theta(\mathcal{S}_{\rm a}) \mu(\theta) d \theta  \label{eq:problem_AoD_approx} \\
                         & \text { s.t. } \;\;\left|\Scal_{\rm a}\right| = 2^{p},
    \end{align}
\end{subequations}
where the objective function~\eqref{eq:problem_AoD_approx} is directly obtained from $\expect_\theta[\tilde{\varepsilon}_{\theta}]$ in \textbf{Lemma~\ref{lemma:Gamma_approximation}}.
Note that $\tilde{\varepsilon}_\theta(\mathcal{S}_{\rm a})$ is a complicated set-valued function of $\mathcal{S}_{\rm a}$, which makes Problem~(P7) intractable due to the lack of a closed-form expression. To tackle this difficulty, we introduce the \emph{Voronoi model}, which is a widely used geometric model for partitioning a space (domain) based on a set of points (samples)~\cite{Voronoi}. 

With the Voronoi model, we can obtain a closed-form expression for $\tilde{\varepsilon}_\theta(\mathcal{S}_{\rm a})$ and reformulate Problem (P7) into a tractable optimization problem. Specifically, for any arbitrary sampling set $\mathcal{S}_{\rm a}=\{\hat{\theta}_i\}$, the Voronoi cell (region) of an arbitrary sample $\hat{\theta}_i$ is defined as
\begin{equation}
    \label{eq:Voronoi_cell}
        \begin{aligned}
            \mathcal{C}_i = \left\{ \theta : \left| \theta-\hat{\theta}_i \right| \le \left| \theta-\hat{\theta}_j \right|,  1\le \forall i \neq j \le 2^p\right\}.
        \end{aligned}
\end{equation}
\begin{figure}[t]
    \centering
    \begin{subfigure}[b]{\linewidth}
        \includegraphics[width=\linewidth]{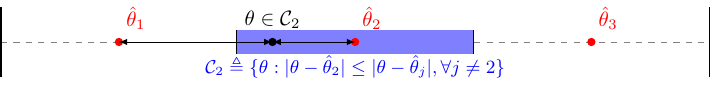}
    \end{subfigure}
    \caption{\small An example of uniform partitioning of the angle domain $\mathcal{Q}$ with $3$ angle samples $\{\hat{\theta}_1,\hat{\theta}_2,\hat{\theta}_3\}$. The blue region is the Voronoi cell $\mathcal{C}_2$. The boundaries of Voronoi cells $\dot{\theta}_i$ are represented by black lines. Red dots denote the angle samples. The black dot denotes an arbitrary angle $\theta$ that falls into the Voronoi cell $\mathcal{C}_2$, whose optimal angle sample is $\hat{\theta}_2$.
    }
    \vspace{-16pt}
    \label{fig:voronoi_diagram}
\end{figure}\noindent Fig.~\ref{fig:voronoi_diagram} provides an example of the Voronoi partitioning of $\mathcal{Q}$ with $3$ angle samples $\mathcal{S}_{\rm a}=\{\hat{\theta}_1,\hat{\theta}_2,\hat{\theta}_3\}$. Mathematically, Voronoi cell $\mathcal{C}_2$ collects all angles $\theta$ that are closer to $\hat{\theta}_2$ than to $\hat{\theta}_1$ and $\hat{\theta}_3$. For example, if $\theta$ falls into $\mathcal{C}_2$, we have $\hat{\theta}_2=\arg\min_{\hat{\theta}\in \mathcal{S}_{\rm a}} \varepsilon_\theta$ being the optimal angle sample. Thus, given $\mathcal{S}_{\rm a}$, the conditional expected angle sampling error given $\theta\in \mathcal{C}_i$ can be obtained as follows.
\begin{lemma}
\label{c:optimal_angle_codeword}
Given the angle sampling set $\mathcal{S}_{\rm a}=\{\hat{\theta}_i\}_{i=1}^{2^p}$, the conditional expected angle sampling error $\tilde{\varepsilon}_\theta(\mathcal{S}_{\rm a})$ given $\theta \in \mathcal{C}_i$ for an arbitrary Voronoi cell $\mathcal{C}_i$ can be obtained as
\begin{equation}
\mathbb{E}_\theta [\tilde{\varepsilon}_\theta(\mathcal{S}_{\rm a})|\mathcal{C}_i] = 
\frac{ 
(\acute{\theta}_{i} - \acute\theta_{i-1})^2 + (\acute{\theta}_{i+1} - \acute{\theta}_{i})^2
}{4(\acute{\theta}_{i+1}-\acute{\theta}_{i-1})},
\end{equation}
where $\acute{\theta}_0=\theta_{\min}, \acute{\theta}_{2^p+1}=\theta_{\max}$, and $\acute{\theta}_i=\hat{\theta}_i$ for $i\in\{1,2,\cdots,2^p+1\}$.
\end{lemma}
\begin{proof}
    From~\eqref{eq:Voronoi_cell}, given $\{\hat{\theta}_i\}$, the Voronoi cells $\{\mathcal{C}_i\}_{i=1}^{2^p}$ can be constructed as $\mathcal{C}_i = [ ({\acute{\theta}_{i-1}+\acute{\theta}_{i}})/{2},({\acute{\theta}_i+\acute{\theta}_{i+1}})/{2})$~\cite{Voronoi}. Thus, given $\theta\in \mathcal{C}_i$, $\mathbb{E}_\theta [\tilde{\varepsilon}_\theta(\mathcal{S}_{\rm a})|\mathcal{C}_i]$ can be obtained as $\int_{\mathcal{C}_i} \tilde{\varepsilon}_\theta(\mathcal{S}_{\rm a})\mu(\theta) d\theta$, where $\tilde{\varepsilon}_\theta(\mathcal{S}_{\rm a}) = \min_{\hat{\theta}\in \mathcal{S}_{\rm a}} |\theta-\hat{\theta}|=|\theta-\acute{\theta}_i|$, leading to the desired result.
\end{proof}\vspace{-0.5em}
Moreover, according to the law of total expectation, the objective function~\eqref{eq:problem_AoD_approx} in Problem (P7) can be equivalently re-expressed as
\begin{equation}
    \label{eq:Voronoi_cell_angle_sampling}
        \int_{\mathcal{Q}} \tilde{\varepsilon}_\theta \mu(\theta) d \theta = \sum_{i=1}^{2^p} \mathbb{E}_\theta[\tilde{\varepsilon}_\theta| \mathcal{C}_i] \cdot \mathrm{Pr}(\theta\in\mathcal{C}_i),
\end{equation}
where $\mathrm{Pr}(\theta\in\mathcal{C}_i)$ is the probability that $\theta$ falls into $\mathcal{C}_i$. Finally, by substituting \eqref{eq:Voronoi_cell_angle_sampling} into \eqref{eq:problem_AoD_approx}, Problem (P7) can be equivalently reformulated as
\begin{subequations}
    \begin{align}
\label{eq:problem_AoD_approx_k}\text{(P8):}\;\; & \min_{\mathcal{S}_{\rm a}}\quad\sum_{i=1}^{2^p} \;\; \expect_{\theta}[\tilde{\varepsilon}_\theta(\mathcal{S}_{\rm a})|\mathcal{C}_i] \cdot \mathrm{Pr}(\theta\in\mathcal{C}_i) \\
                                                          & \; \text{s.t.} \quad\;\; |\mathcal{S}_{\rm a}| = 2^p.  \label{eq:C1}
    \end{align}
\end{subequations}
Problem (P8) now becomes a convex optimization problem, whose optimal solution is given below, with the proof provided in Appendix~\ref{sec:appendix_angle_sampling}.
        
\begin{framed}
    \vspace{-4pt}
    {\setlength\abovedisplayskip{0pt}
        \setlength\belowdisplayskip{0pt}
        \begin{equation}
            \label{eq:nf_aod_samples}
            {\begin{aligned}
                     & \textbf{Near-field Angle Samples:}                                                                       \\
                     & \Scal_{\rm a}^* =\left[\hat{\theta}^*_1,\hat{\theta}^*_2,\cdots,\hat{\theta}^*_{2^p}\right],   \text{with } \\
                     & \hat{\theta}^*_i = \theta_{\mathrm{min}}+i\cdot\frac{D(\mathcal{Q})}{2^{p}+1}, i=1,2,\dots,2^p,
                \end{aligned}}
        \end{equation}}
\end{framed}\noindent
where $D(\mathcal{X}) = \max_{x\in\mathcal{X}} x - \min_{x\in\mathcal{X}} x$ is the \textit{difference between the maximum and minimum values} of an arbitrary real-valued set $\mathcal{X}$.
Note that the optimal angle samples $\{\hat{\theta}^*_i\}$ are uniformly distributed in the angle domain $\mathcal{Q}$, which is consistent with the widely-considered DFT codebook for spatial angle sampling in far-field communications~\cite{Cui2023,2023sparsearray}.

{\textbf{Proposed range sampling method:}}
\label{sec:minimization_epsL}
Similarly, the subproblem of minimizing $\expect_r[\tilde{\varepsilon}_{r}]$ can be formulated as
\begin{subequations}
    \begin{align}
        \label{eq:problem_R_approx}\text{(P9):}\;\; & \min_{\Scal_{\rm r}} \;\; \int_{\mathcal{R}} \tilde{\varepsilon}_{r}(\Scal_{\rm r}) \mu(r) d r \\
                                                    & \text { s.t. }\;\; \left|\Scal_{\rm r}\right| = 2^{q}. \label{eq:P9C1}
    \end{align}
\end{subequations}
In Problem (P9), since $\tilde{\varepsilon}_{r}(\Scal_{\rm r})=\min_{\hat{r}\in\mathcal{S}_{\rm r}}\left|\frac{1}{\hat{r}}-\frac{1}{r}\right|$ is a non-linear function of $r$, the relationship between the Voronoi cell and the corresponding range sample is even more complicated. To overcome this difficulty, we first partition $\mathcal{R}$ into $2^q$ intervals (cells) $\{\mathcal{I}_i\}_{i=1}^{2^q}$ (and $2^q$ corresponding samples $\{\hat{r}_i\}$ with $\hat{r}_i \in \mathcal{I}_i,\forall i$) subject to the conditions that $\{\mathcal{I}_i\}$ are non-overlapping and their union equals $\mathcal{R}$, i.e., $\cup_{i=1}^{2^q} \mathcal{I}_i = \mathcal{R}$.
We then define $\mathring{\varepsilon}_{r}(\hat{r}_i) \triangleq \left|{1}/{r}-{1}/{\hat{r}_i}\right|$ as the \emph{surrogate} sampling error of an arbitrary range sample $\hat{r}_i$, based on which we obtain the following lemma.
\begin{lemma}
    \label{l:new_sample_feedback_model_eps}
    \label{eq:problem_AoD_approx_reformulate}
    Given any non-overlapping range partitioning $\{\mathcal{I}_i\}$ whose union equals $\mathcal{R}$, the expected sampling error related to the range $\expect_r[\tilde{\varepsilon}_{r}]$ in \textbf{Lemma~\ref{lemma:Gamma_approximation}} can be upper-bounded as
    \begin{equation}
        \label{eq:upper_bound_sampling_error}
        \int_{\mathcal{R}} \tilde{\varepsilon}_{r}(\Scal_{\rm r}) \mu(r) d r\le\sum_{i=1}^{2^q} \expect_r[\mathring{{\varepsilon}}_{r}(\hat{r}_i)|\mathcal{I}_i] \mathrm{Pr}(r\in\mathcal{I}_i),
    \end{equation}
    where $\expect_r[\mathring{{\varepsilon}}_{r}(\hat{r}_i)|\mathcal{I}_i]=\int_{\mathcal{I}_i} \mathring{{\varepsilon}}_{r} \mu(r) d r$ is the conditional expected range sampling error given $r\in\mathcal{I}_i$. \noindent
\end{lemma}
    \begin{proof}
    According to the law of total expectation, we can expand $\int_{\mathcal{R}}\! \tilde{\varepsilon}_{r} \mu(r) d r\! =\! \sum_{i=1}^{2^q} \int_{\mathcal{I}_i}\!\! \tilde{\varepsilon}_{r}\mu(r) d r \mathrm{Pr}(r\in \mathcal{I}_i)$ for any feasible partitioning $\{\mathcal{I}_i\}$. 
    Specifically, given $r\in\mathcal{I}_i$, we have $\int_{\mathcal{I}_i} \tilde{\varepsilon}_{r} \mu(r) d r = \int_{\mathcal{I}_i} \min_{\hat{r}\in\Scal_{\rm r}}\left| \frac{1}{r}-\frac{1}{\hat{r}} \right| \mu(r) d r$, with $ \tilde{\varepsilon}_{r}= \min_{\hat{r}\in \Scal_{\rm r}}\left| \frac{1}{r}-\frac{1}{\hat{r}} \right| \le \mathring{{\varepsilon}}_{r} (\hat{r}_i)= \left| \frac{1}{r}-\frac{1}{\hat{r}_i} \right|$, which directly leads to $\int_{\mathcal{I}_i} \tilde{\varepsilon}_{r} \mu(r) d r \le \int_{\mathcal{I}_i} \mathring{{\varepsilon}}_{r}(\hat{r}_i) \mu(r) d r = \expect_r[\mathring{{\varepsilon}}_{r}(\hat{r}_i)|\mathcal{I}_i]$. Thus, by collecting this inequality for all $i$, we have the desired result.
\end{proof}\vspace{-0.5em}

Note that $\mathbb{E}_r[\mathring{\varepsilon}_r(\hat{r}_i)|\mathcal{I}_i] \mathrm{Pr}(r\in \mathcal{I}_i)$ in~\eqref{eq:upper_bound_sampling_error} is only dependent on the range cell $\mathcal{I}_i$ and the corresponding sample $\hat{r}_i$, given any feasible partitioning $\{\mathcal{I}_i\}$ (non-overlapping and constituting $\mathcal{R}$). Thus, by adopting \textbf{Lemma~\ref{l:new_sample_feedback_model_eps}}, we can approximately solve Problem (P9) by minimizing the upper-bound in~\eqref{eq:upper_bound_sampling_error} as follows
\begin{subequations}
    \begin{align}
        \label{eq:problem_r_approx_k}\text{(P10):}\;\; & \min_{\{\mathcal{I}_i\},\{\hat{r}_i\}}\sum_{i=1}^{2^q} \;\; \expect_r[\mathring{{\varepsilon}}_{r}(\hat{r}_i)|r\in \mathcal{I}_i] \mathrm{Pr}(r\in\mathcal{I}_i) \\
                                                        & \;\;\;\; \text{s.t.} \;\;\;\; \mathcal{I}_i \cap \mathcal{I}_j = \emptyset, \forall i \ne j, \quad  \label{eq:C4}\\
                                                        & \phantom{\;\;\; \text{s.t.} \;\;\;\;} \cup_{i=1}^{2^q} \mathcal{I}_i = \mathcal{R}, \quad   \label{eq:C5}                 \\
                                                        &\phantom{\;\;\; \text{s.t.} \;\;\;\;\;} \hat{r}_i \in \mathcal{I}_i.\label{eq:C6}
    \end{align}
\end{subequations} 
Note that the constraints \eqref{eq:C4}--\eqref{eq:C5} ensure that the partitioning $\{\mathcal{I}_i\}$ is feasible, i.e., non-overlapping and complete, and the constraint \eqref{eq:C6} is to ensure that the range sample $\hat{r}_i$ is in the corresponding range cell $\mathcal{I}_i$.
To solve this problem, we first optimize the range samples $\hat{r}_i$ in each cell $\mathcal{I}_i$ given fixed partitioning $\{\mathcal{I}_i \}$, and then optimize the partitioning $\{\mathcal{I}_i \}$ with optimized range samples $\{\hat{r}_i\}$.

\underline{{{\textit{Optimizing $\{\hat{r}_i\}$ given $\{\mathcal{I}_i\}$}:}}}
Given any feasible partitioning $\{\mathcal{I}_i\}$ that satisfies constraints \eqref{eq:C4}--\eqref{eq:C5}, Problem (P10) reduces to the following subproblem for optimizing the range sample $\hat{r}_i$ in each cell $\mathcal{I}_i$
\begin{subequations}
    \begin{align}
        \label{eq:problem_r_approx_I_i}\text{(P11):}\;\; & \min_{\hat{r}_i} \expect_r \left[\mathring{\varepsilon}_{r} (\hat{r}_i) | \; r\in \mathcal{I}_i\right] \;\; \\
                                                              & \;\text{s.t.} \;\; \text{~\eqref{eq:C6}}. \notag
    \end{align}
\end{subequations}
\begin{proposition}
    \label{p:optimal_range_sample}
    The optimal solution to Problem (P11) is
    \begin{equation}
        \label{eq:optimal_range_sample}
        \hat{r}_i^{*} = \frac{1}{2}\left(\dot{r}_{i-1}+\dot{r}_{i}\right),
    \end{equation}
    where $\dot{r}_{i-1} \triangleq \min_{r \in \mathcal{I}_i} r$ and $\dot{r}_{i} \triangleq \max_{r \in \mathcal{I}_i} r$. 
    Moreover, the corresponding objective value in~\eqref{eq:problem_r_approx_I_i} is 
    $$
    \expect_r \left[\mathring{\varepsilon}_{r}(\mathcal{I}_i,\hat{r}^*) | \; r\in \mathcal{I}_i\right]=\frac{2}{D(\mathcal{I}_i)}\left(\log \frac{\sqrt{\xi_i}+1/\sqrt{\xi_{i}}}{2}\right),
    $$
    where $\xi_i =\sqrt{\frac{\dot{r}_{i}}{\dot{r}_{i-1}}}$, and $D(\mathcal{I}_i) = \dot{r}_{i}-\dot{r}_{i-1}$ is the difference between the maximum and minimum of $\mathcal{I}_i$.
\end{proposition}
\begin{proof}
    See Appendix~\ref{sec:appendix_range_sampling_cells}.
\end{proof}

\underline{{{\textit{Optimizing $\{\mathcal{I}_i\}$ given $\{\hat{r}_i^*\}$}:}}} Substituting the obtained $\hat{r}_i^*$ in~\eqref{eq:optimal_range_sample} into~\eqref{eq:problem_r_approx_k}, Problem (P10) reduces to the following subproblem for optimizing the partitioning $\{\mathcal{I}_i\}$
\begin{equation*}
    \label{eq:problem_rA_approx_D}
    \begin{aligned}
        \text{(P12):}\;\; & \min_{\{\mathcal{I}_i\}}\;\;\sum_{i=1}^{2^q} \frac{2 }{D(\mathcal{I}_i)}\!\left(\log \frac{{\xi_i}+1/{\xi_{i}}}{2}\right) \!\mathrm{Pr}(r\in \mathcal{I}_i)\;\;
        \\& \;\text{s.t.} \;\; \text{~\eqref{eq:C4},\eqref{eq:C5}},
    \end{aligned}
\end{equation*}
which is optimally solved below.
\begin{proposition}
    \label{p:optimal_range_divide}
    The optimal solution to Problem (P12) is
    \begin{equation}
        \mathcal{I}_i^* = \left[r_{\rm min} \xi^{\frac{i-1}{2^q}},r_{\rm min}\xi^{\frac{i}{2^q}}\right),\forall i\in\{1,2,\cdots,2^{q}\}, 
    \end{equation}
    where $\xi= \sqrt{\frac{r_{\rm max}}{r_{\rm min}}}$, and the optimized objective is
    \begin{equation}
        \label{eq:optimal_range_error_sum}
       {\expect_r[\mathring{{\varepsilon}}_{r}(\mathcal{I}_i^*)] = \frac{2^{q}}{D(\mathcal{R})}\left(\log \frac{\sqrt[2^{q}]{\xi}+1/\sqrt[2^{q}]{\xi}}{2}\right).}
    \end{equation}

\end{proposition}
    \begin{proof}
        See Appendix~\ref{sec:appendix_range_sampling}.
    \end{proof}\vspace{-0.5em}
Based on the \textbf{Propositions~\ref{p:optimal_range_sample}} and \textbf{\ref{p:optimal_range_divide}}, and defining $\xi_{\Sigma} = ({r_{\max}}/{r_{\min}})^{1/2^q}$, we obtain a suboptimal solution to Problem (P9) as follows
\begin{framed}
    {\setlength\abovedisplayskip{0pt}
        \setlength\belowdisplayskip{0pt}
        \begin{equation}
            \label{eq:nf_r_samples}
            \begin{aligned}
                 & \textbf{Near-field Range Samples:}                                                                            \\
                 & \Scal_{\rm r}^* = \left[ \hat{r}_1,\hat{r}_2,\cdots,\hat{r}_{2^q}\right],    \text{where }                          \\
                 & \hat{r}_i^* = \frac{r_{\min}}{2} (\xi_{\Sigma}^{-1}+1){\xi^{i}_{\Sigma}}, i =1,2,\dots,2^q.
            \end{aligned}
        \end{equation}}
\end{framed}
The proposed range sampling scheme in~\eqref{eq:nf_r_samples} is referred to as \textit{geometric sampling}, since the range samples are generated by a geometric sequence (i.e., the ratio of two adjacent samples ${\hat{r}_i^*}/{\hat{r}_{i-1}^*}= \xi_{\Sigma}$ is a constant), which assigns more samples at closer ranges. The corresponding codebook $\Bcal_1=\{\mathbf{b}_1(i,j)=\mathbf{a}(\hat{\theta}_i,\hat{r}_j)\mid \hat{\theta}_i\in \Scal_{\rm a}^*,\hat{r}_j \in \Scal_{\rm r}^* \}$ is called the \textit{geometric codebook}.

\begin{figure}[t]
\centering
\begin{minipage}{0.7\linewidth}
    \includegraphics[width=\linewidth]{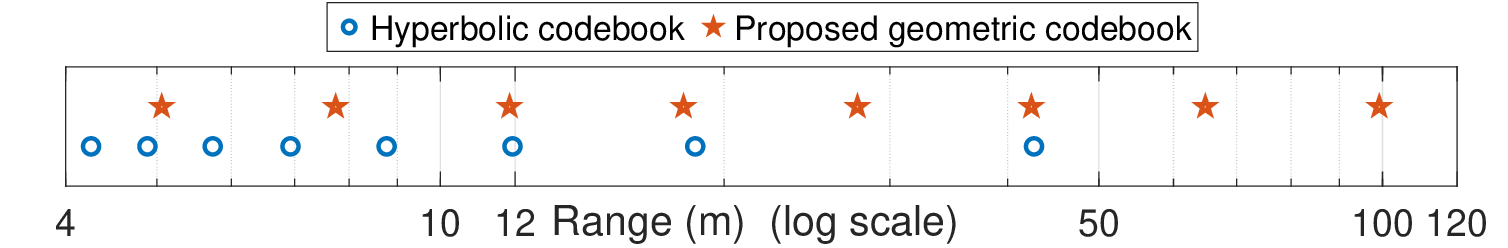}
    \subcaption{Range samples.}
\end{minipage}
\hfill
\begin{minipage}{0.7\linewidth}
    \includegraphics[width=\linewidth]{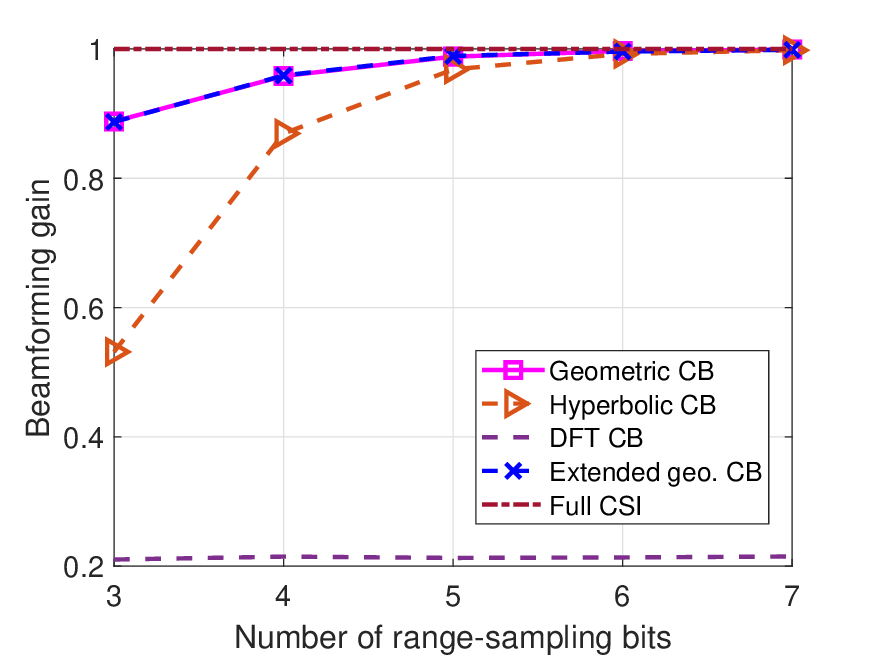}
    \subcaption{Beamforming gain vs. the number of range-sampling bits $q$.}
\end{minipage}
\caption{\small
Comparison between geometric and hyperbolic codebooks. For other simulation parameters, please refer to Section~\ref{sec:number_sim}.}
\label{fig:uniform_group}
\vspace{-1.2em}
\end{figure}

\begin{remark}[Geometric versus (vs.) hyperbolic]
\label{r:scheme_compare}
The proposed \textit{geometric} codebook is a new polar-domain codebook and is different from the \textit{hyperbolic codebook} proposed in~\cite{cuiChannelEstimationExtremely2022}. It is worth noting that the range samples of the hyperbolic codebook are sampled over a bounded range $[r_{\min}, r_{\max}]$ as
\begin{equation}
    \label{eq:nf_r_samples_polar}
\hat r^{\rm hy}_i=\frac{2^{q} r_{\max}}{i(\xi^{2}-1)+2^{q}},~ \forall i=1,\dots,2^{q}.
\end{equation}
As illustrated in Fig.~\ref{fig:uniform_group}(a), the hyperbolic codebook assigns too few samples in the long range regime. 
For example, it only assigns two range samples for the regime of $[12,120]$ meter (m), while there are $92\%$ of the users located in this regime under uniform distribution. This suggests that the hyperbolic codebook fails to adapt to the user range distribution; while the proposed geometric range sampling assigns more samples to cover users in the long range regime, thus achieving a higher beamforming gain, as shown in Fig.~\ref{fig:uniform_group}(b).
\end{remark}

\begin{figure}[t]
\centering
\begin{minipage}{0.7\linewidth}
    \includegraphics[width=\linewidth]{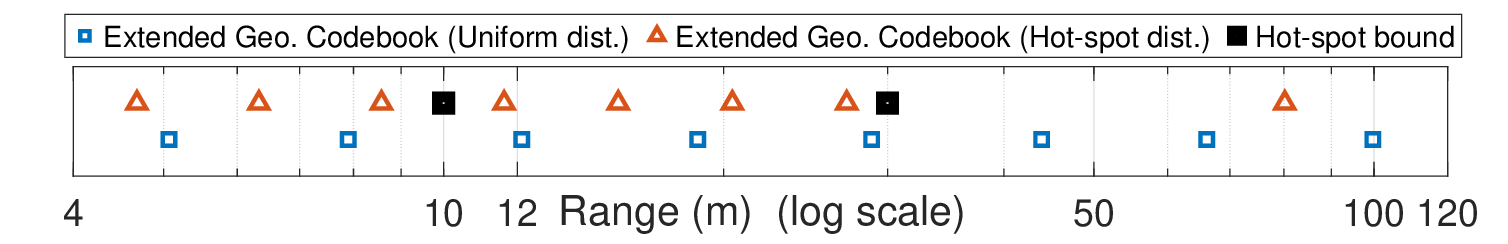}
    \subcaption{\small Range samples. \textit{Hot-spot bound} indicates the boundaries of the hot-spot user region.}
\end{minipage}
\hfill
\begin{minipage}{0.7\linewidth}
    \includegraphics[width=\linewidth]{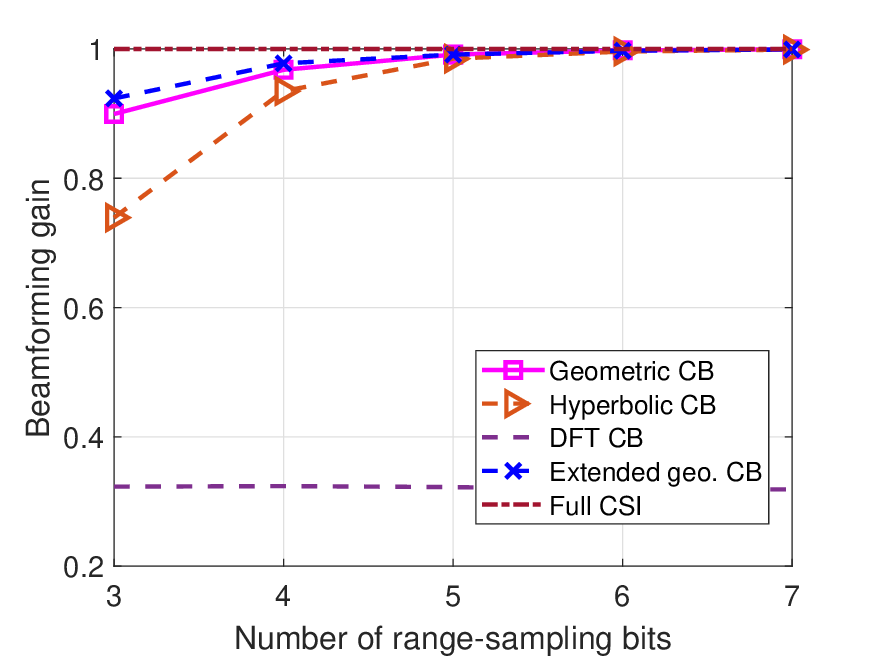}
    \subcaption{Beamforming gain vs. the number of range-sampling bits $q$.}
\end{minipage}
\caption{\small
Comparison between uniform and non-uniform (hot-spot\protect\footnotemark) user distributions. For other simulation parameters, please refer to Section~\ref{sec:number_sim}.}
\label{fig:nonuniform_group}
\vspace{-1.2em}
\end{figure}

\footnotetext{Hot-spot distribution: $r$ has a probability of $90\%$ to uniformly distribute in $[10,20]$ and $10\%$ to uniformly distribute in $[4,10]\cup [20,120]$.} 
\subsubsection{\textbf{Proposed Solution to Problem (P6)}}
\label{sec:adptive_allocation_I}
After solving Problem (P5), we optimize the bit allocation $(p,q)$ to solve Problem (P6). However, it is generally difficult to solve due to non-convex objective function $\tilde{\Gamma}(p,q)$ and the integer constraints on the feedback bit allocation $(p,q)$ in~\eqref{eq:P4_cardinality_cons_pq}. To tackle this difficulty, we propose an exhaustive search algorithm to obtain its optimal solution, where we approximate the objective function $\tilde{\Gamma}(p,q)$ using Monte-Carlo simulations. 
Since for practical implementation, $B_1$ is generally small (e.g., $B_1$ being 16--20 can achieve satisfactory rate performance, see Fig.~\ref{fig:bits_allocation}), which renders computational complexity of the exhaustive search acceptable. 
The overall computational complexity of the proposed exhaustive search algorithm is on the order of $\mathcal{O}(2^{B_1} N_{\rm mc})$ with $N_{\rm mc}$ being the number of Monte-Carlo simulations. For practical implementation, $N_{\rm mc}$ can be set to a moderate value (e.g., $N_{\rm mc}=300$) \cite{cuiChannelEstimationExtremely2022}.
\begin{figure}[t]
\centering
\begin{minipage}{0.7\linewidth}
    \includegraphics[width=\linewidth]{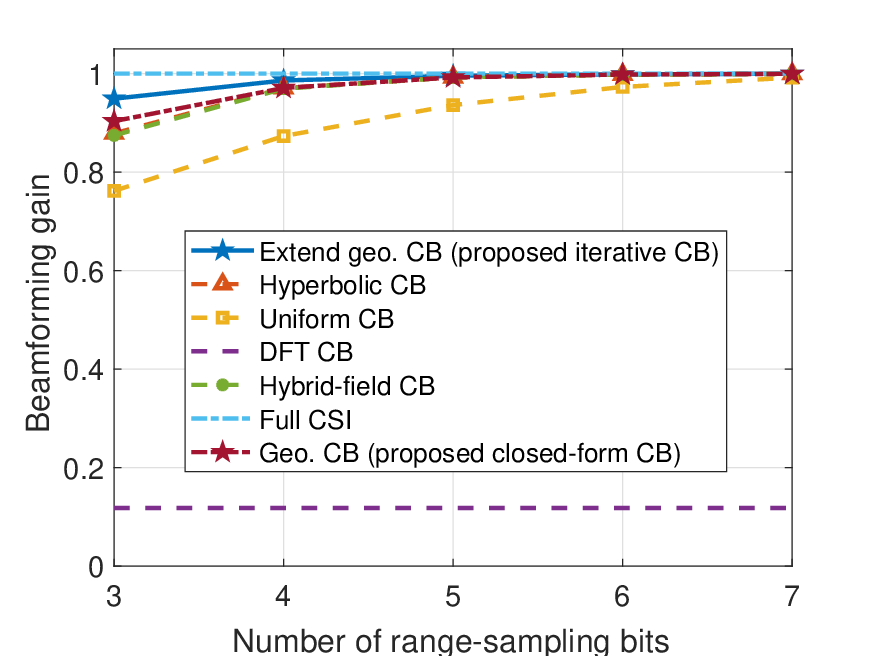}
    \subcaption{\small Gaussian.}
\end{minipage}
\hfill
\begin{minipage}{0.7\linewidth}
    \includegraphics[width=\linewidth]{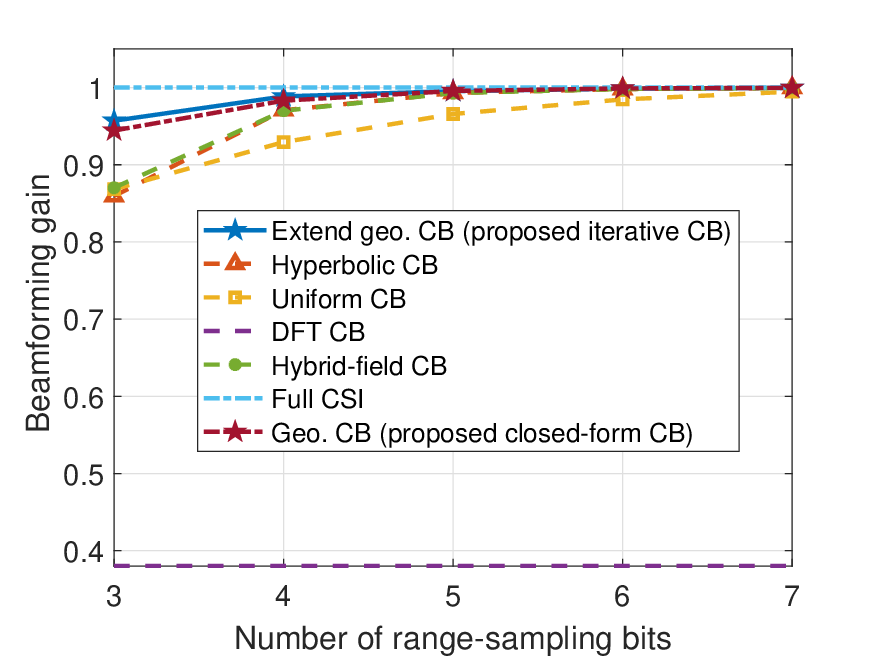}
    \subcaption{\small Mixture of Gaussian.}
\end{minipage}
\caption{\small
Beamforming performance under Gaussian and mixture of Gaussian user distributions. For other simulation parameters, please refer to Section~\ref{sec:number_sim}.}
\label{fig:nonuniform_group_oth}
\vspace{-0.8em}
\end{figure}

\begin{remark}[Non-uniform user distributions]
    \label{r:non_unif}
    The geometric codebook under more general non-uniform user distributions can be obtained by using the alternating-based algorithm similar to~\cite{llyodLeastSquaresQuantization1982}, which is referred to as the \emph{extended geometric codebook}. {Specifically, suppose that we have $N_{\rm sam}$ samples of user-range data, $\{r_n\}_{n=1}^{N_{\rm sam}}$, which are drawn from a non-uniform distribution $\mu(r)$.
    As summarized in Algorithm 1, the extended geometric range-sampling algorithm includes two main steps: 1) updating the range codewords $\{\hat{r}_i\}$ based on the current range-cell partitioning $\{\mathcal{I}_i\}$ (lines 3--6); and 2) updating the range-cell partitioning $\{\mathcal{I}_i\}$ based on the updated range codewords $\{\hat{r}_i\}$ (lines 7--10).
    Note that, if the user distribution is uniform, the above alternating algorithm for the non-uniform distribution essentially reduces to the same range sampling as in \eqref{eq:nf_r_samples} (see Fig.~\ref{fig:nonuniform_group}(a)}).
    The angle-sampling method in~\eqref{eq:nf_aod_samples} can be similarly extended to non-uniform distributions by replacing the objective function of line 6 in Algorithm~\ref{a:location_distribution} with that in Problem (P8) (i.e., minimizing $\expect_\theta[{\varepsilon}_\theta(\hat{\theta}_i)|\mathcal{C}_i]$). 
    It is worth noting that the closed-form geometric codebook (designed based on uniform user distribution) can still provide satisfactory beamforming performance under non-uniform distributions, as shown in Fig.~\ref{fig:nonuniform_group}(b). 
    This is because the geometric allocation of range samples spans both the near and long range regimes, thereby offering robustness to diverse user distributions. 
    {
    Moreover, for other non-uniform distributions (e.g., Gaussian distribution and mixture of Gaussian distribution), the extended geometric codebook can also be applied to further enhance the beamforming performance, as shown in Fig.~\ref{fig:nonuniform_group_oth}.}
\end{remark}

\begin{algorithm}[t]{
    \caption{Extended geometric range sampling}
\label{alg:four_step_lloyd_r_opt}\begin{algorithmic}[1]
    \setlength{\itemsep}{0pt}
\setlength{\parskip}{0pt}
\Input User-range data $\{r_{n}\}_{n=1}^{N_{\rm sam}}$, bits $q$, tolerance $\varepsilon$
\State Randomly initialize $\{\mathcal{I}_i\}_{i=1}^{2^q}$.
\Repeat
    \For{$i=1$ to $2^q$}
        \State $\mathcal{R}_{i} \gets \{ r_{n} : n \in \mathcal{I}_{i} \}$
        \State $\hat{r}_{i} \gets \arg\min_{c>0}\, \sum_{r \in \mathcal{R}_{i}} \mathring{\varepsilon}_r (c)$ //  Update the range codeword.
    \EndFor
    \For{$i=1$ to $2^q$}
    \State $\mathcal{I}_{i} \gets \{ n : i = \arg\min_{j} |1/r_{n} - 1/\hat{r}_{j}| \}$ //  Update the range-cell partitioning.
    \EndFor
    \State $\delta \gets \max_{i} \left| \hat{r}_{i} - \hat{r}_{i}^{\text{prev}} \right|$, store $\hat{r}_{i}^{\text{prev}} \gets \hat{r}_{i}$.
\Until{$\delta < \varepsilon$}
\State \textbf{Return} sorted $\{ \hat r_{i} \}_{i=1}^{2^q}$.
\end{algorithmic}}
\end{algorithm}

\subsection{Codebook Design for Phase 2}
\label{sec:effective_channel_feedback}

After Phase 1, the BS further designs its digital beamforming matrix $\mathbf{F}_{\mathrm{BB}}$ in Phase~2. Specifically, the effective channel matrix for the users is given by $\mathbf{G} = \left[\mathbf{g}_1,\mathbf{g}_2,\cdots,\mathbf{g}_K\right]$ with $\mathbf{g}_k^H = \mathbf{h}_k^H \mathbf{F}_{\mathrm{RF}}$. When $ \hat{\theta}^\star\to\theta$ and $ \hat{r}^\star\to r$, the diagonal elements of $\mathbf{G}$ can be obtained as
\begin{equation}
    \label{eq:effective_channel_diag}
\!\![\mathbf{G}]_{k,k} \!= \! \mathbf{h}_k^H \mathbf{f}_{\mathrm{RF},k} \!=\!  \beta_{k} \mathbf{a}^H({\theta}_{k},{r}_{k}) \mathbf{a}(\hat{\theta}^\star_{k},\hat{r}^\star_{k}) \approx\beta_k, \forall k,
\end{equation}
where the approximation is due to $\mathbf{a}({\theta}_{k},{r}_{k})^H\mathbf{a}(\hat{\theta}^\star_{k},\hat{r}^\star_{k}) \to 1$. Similarly, the non-diagonal elements of $\mathbf{G}$ can be approximated as
\begin{equation}
    \label{eq:effective_channel_non_diag}
\!\![\mathbf{G}]_{k,v} \!=\!  \mathbf{h}_k^H \mathbf{f}_{\mathrm{RF},v} \!\!=\!\!\beta_{k} \mathbf{a}^H(\theta_{k},r_{k})\mathbf{a}(\hat{\theta}^\star_{v},\hat{r}^\star_{v})\approx 0, \forall k\ne v,
\end{equation}
when $\hat{\theta}_k \ne \hat{\theta}^\star_{v}$. One can observe from \eqref{eq:effective_channel_diag} and \eqref{eq:effective_channel_non_diag} that the near-field beam-focusing effect can help eliminate IUI. However, when the number of users is large (or the users are close to each other), the residual IUI may still cause a degraded sum-rate, thus necessitating explicit feedback of the effective channels $\mathbf{G}$ to help the BS design IUI-eliminating digital beamforming $\fbb$. 

To this end, we adopt the RVQ codebook to generate a feedback codebook for $\mathbf{G}$. Specifically, each codeword of $\Bcal_2$, i.e., $\mathbf{b}_2{(i)}$, is generated according to the distribution of $\mathbf{g}$ given in~\eqref{eq:effective_channel_diag} and~\eqref{eq:effective_channel_non_diag} by Monte-Carlo simulations. Note that RVQ is easy to implement and is \emph{asymptotically} optimal when the number of users $K$ is sufficiently large~\cite{au-yeungPerformanceRandomVector2007a}. 
With the feedback in Phase 2 of all the $K$ users $\{\hat{\mathbf{g}}_k\}$, the BS can then utilize the ZF technique as a suboptimal yet efficient digital beamforming design to solve Problem (P3), which has been widely used in the literature (see, e.g., \cite{limitedfeedbackSurvey,alkhateebLimitedFeedbackHybrid2015,lywTut}). Mathematically, the ZF-based digital beamforming for user $k$ can be expressed as 
$
\fbbk{k}= {[(\hat{\mathbf{G}}^H \hat{\mathbf{G}})^{-1} \hat{\mathbf{G}}]_{(:,k)} \over \|[(\hat{\mathbf{G}}^H \hat{\mathbf{G}})^{-1} \hat{\mathbf{G}}]_{(:,k)}\|_2},\forall k,
$
where $\hat{\mathbf{G}} = [\hat{\mathbf{g}}_1,\hat{\mathbf{g}}_2,\cdots,\hat{\mathbf{g}}_K]$ is the effective channel matrix given the feedback $\{\hat{\mathbf{g}}_k\}$. 


\begin{remark}[Multi-path channels]
    \label{r:multi-path_channel}
    The proposed scheme can be extended to multi-path channels by feeding back the parameters (angle, range, complex gain) for each path. Specifically, the geometric codebook can be used to feed back the angle and range of each path, while an RVQ codebook can be used for the feedback of the complex-valued channel gains of all paths after they are estimated by the user. However, the optimal bit allocation among paths and a tailored feedback framework for the hybrid beamforming architecture still remain open problems.
\end{remark}
\begin{remark}[Imperfect CSI]
    \label{r:imperfect_csi}
    The proposed framework can be extended to handle imperfect CSI at the user side. With channel estimation errors, the true channel $\mathbf{h}_k$ can be modeled to lie within an uncertainty set around the estimated channel $\hat{\mathbf{h}}_k$~\cite{robust}. A robust feedback strategy can then be adopted, where the user selects the codeword that maximizes the worst-case received power over this uncertainty set~\cite{robust}, which ensures a guaranteed rate performance level despite CSI inaccuracy. However, designing a more dedicated and efficient feedback mechanism for imperfect CSI is left for future work.
\end{remark}

\begin{figure}[t]
    \centering
    \begin{subfigure}{0.485\linewidth}
        \centering
        \includegraphics[width=\linewidth]{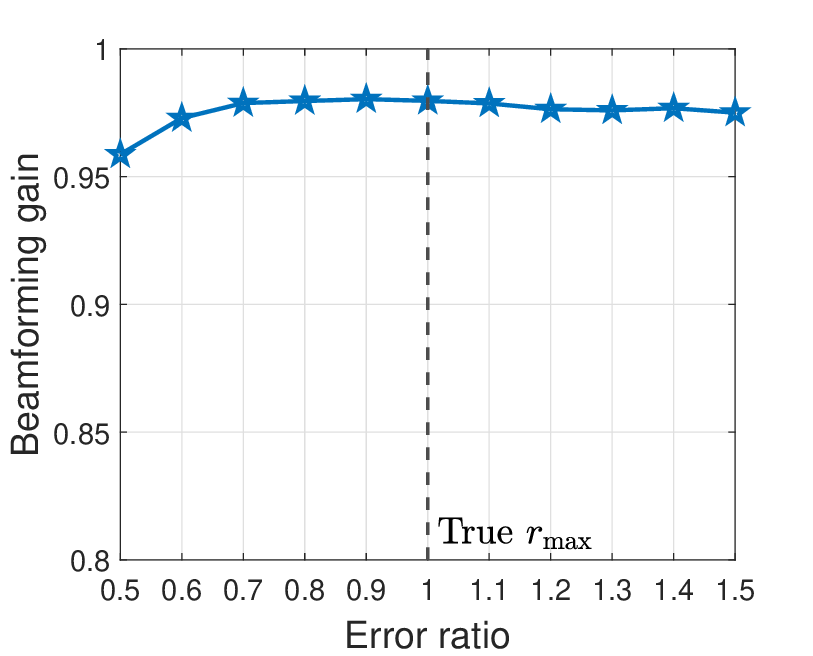}
        \subcaption{Mismatch in user distribution parameters. The error ratio is defined as $r_{\max}^{\text{est.}}/r_{\max}^{\text{true}}$, where $r_{\max}^{\text{est.}}$ and $r_{\max}^{\text{true}}$ denote the estimated and true maximum user ranges, respectively.}
    \end{subfigure}
    \begin{subfigure}{0.485\linewidth}
        \centering
        \includegraphics[width=\linewidth]{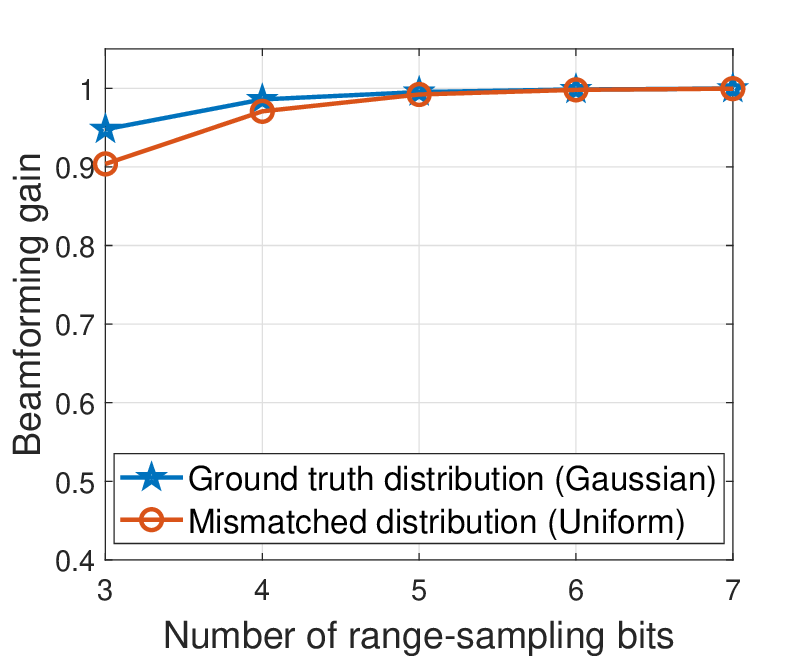}
        \subcaption{Mismatch in user distribution types. Specifically, two codebooks are designed based on uniform and Gaussian distribution, respectively, while the actual user distribution is Gaussian.}
    \end{subfigure}
    \caption{\small The impact of inaccurate user distribution information on the beamforming gain performance of the proposed geometric codebook. For other simulation parameters, please refer to Section~\ref{sec:number_sim}.}
    \label{fig:dist_sensitivity}
\end{figure}

{
\begin{remark}[Other practical considerations]
First, for practical deployments, the BS can estimate the actual user distribution by using, e.g., historical data, global positioning systems (GPS) information, integrated sensing and communication (ISAC) methods, and existing CSI estimation methods \cite{3gpp.37.320,limitedfeedbackSurvey,you2024generationadvancedtransceivertechnologies,cuiChannelEstimationExtremely2022,shen0218channelfeedback,lu2023nearfieldestimation}. Second, the codebook sharing between the BS and the users can be performed by following similar procedures defined in model transfer process in 3GPP Rel-18\cite{3gpprel18}.
Third, the signaling overhead and computational complexity associated with codebook updates are generally negligible as the codebook only needs to be updated per serveral hours or days.
Furthermore, to evaluate the impact of inaccurate distribution knowledge on the proposed codebooks, we evaluate the beamforming gain under two scenarios: (1) estimation errors in $r_{\max}$ in Fig.~\ref{fig:dist_sensitivity}(a), and (2) mismatches in distribution type in Fig.~\ref{fig:dist_sensitivity}(b). 
The results indicate that both types of inaccuracies result in only marginal performance degradation, thereby verifying the robustness of the proposed codebook design against imperfect distribution information.
\end{remark}
}

\section{Performance Analysis}
In this section, we analyze the required feedback bits for the proposed three-phase transmission protocol to achieve a target beamforming gain and a target rate gap. For ease of exposition, we consider a uniform distribution of user angles and ranges in this section, while the results can be extended to non-uniform distributions.
\vspace{-0.9em}
\label{sec:analysis}
\subsection{Required Feedback Bits in Phase 1}
\label{sec:bits_analysis_phase1}
We analyze the impact of the number of antennas on the required feedback bits, i.e., the feedback bit overhead for an XL-array equipped with $M$ antennas to achieve a target beamforming gain $\Gamma_0$. Due to the intricate coupling between $\tilde{\Gamma}(p,q)$ and the complicated form of $\mathbb{E}_\theta[\tilde{\varepsilon}_\theta]$ and $\mathbb{E}_r[\tilde{\varepsilon}_{r}]$, we first characterize the expected beamforming gain $\Gamma$ of the XL-array with $M$ antennas as follows
\begin{align}
    \Gamma &\approx \frac{2}{M} \left| \sum_{m=0}^{(M-1)/2} e^{\jmath m \pi \mathbb{E}_\theta[\tilde{\varepsilon}_\theta]} e^{\jmath k_c m^2 d_0^2  \frac{1}{2} \check{\vartheta} \mathbb{E}_r[\tilde{\varepsilon}_{r}]} \right| \notag\\
 & \triangleq \check{f}( \mathbb{E}_\theta [\tilde{\varepsilon}_\theta],\check{\vartheta}\mathbb{E}_r[\tilde{\varepsilon}_{r}],M),\label{eq:approximation_M_E}
\end{align}
where we adopt the approximation in \textbf{Lemma~\ref{lemma:Gamma_approximation}}. 
We assume that the target beamforming gain is sufficiently large (i.e., sufficiently small $\tilde{\varepsilon}_\theta$ and $\tilde{\varepsilon}_{r}$). 
We reveal how the number of angle-sampling bits should grow with the number of antennas $M$, provided that range sampling is sufficiently accurate (e.g., $q>3$ as revealed in Section~\ref{sec:number_sim}).

\begin{lemma}
    \label{c:epc_q}
 For an $M$-antenna XL-array, to achieve a target expected beamforming gain $\Gamma_0\gg 0$ in~\eqref{eq:approximation_M_E}, if the range sampling is sufficiently accurate, the required number of angle-sampling bits $p$ is approximated by
    \begin{equation}
        \begin{aligned}
 \label{eq:AoD_bits}
 p \approx \log_2 {M} - \log_2{D(\mathcal{Q})}-\log_2{M_0 \tilde{\varepsilon}_{\theta,\Gamma_0}}-2,
        \end{aligned}
    \end{equation}
 where $\tilde{\varepsilon}_{\theta,\Gamma_0}$, $M_0$ are constants satisfying $\check{f}(\tilde{\varepsilon}_{\theta,\Gamma_0},0,M_0) = \Gamma_0$.
\end{lemma}  
\begin{proof}
    {See Appendix~\ref{sec:appendix_lemma5}.}
\end{proof}

Specifically, by fixing $M_0$, $\tilde{\varepsilon}_{\theta,\Gamma_0}$ and $D(\mathcal{Q})$ as constants, the target beamforming gain $\check{f}(\tilde{\varepsilon}_{\theta,\Gamma_0},0,M_0) = \Gamma_0$ is fixed, and the required number of angle-sampling bits $p$ is on the order of $\mathcal{O}(\log_2 M)$. 
\begin{lemma}
    \label{c:epc_p}
    For an $M$-antenna XL-array, to achieve a target expected beamforming gain $\Gamma_0\gg 0$ in~\eqref{eq:approximation_M_E}, when angle-sampling is accurate, the required number of range sampling bits $q$ is approximated by
    \begin{equation}
        \label{eq:range_bits}
        \begin{split}
            q &\approx 2\log_2{M} + \log_2 \left(\frac{\log^2 \xi}{D(\mathcal{R})}\right) - \log_2\left({\check{M}_0^2 \tilde{\varepsilon}_{r,\Gamma_0} }\right)\\& - (1 + \log_2(\log 2)),
        \end{split}
    \end{equation}
    where $\tilde{\varepsilon}_{r,\Gamma_0}$, $\check{M}_0$ are constants satisfying $\check{f}(0,\check{\vartheta} \tilde{\varepsilon}_{r,\Gamma_0},\check{M}_0) = \Gamma_0$, and $\log(\cdot)$ is the natural logarithm.
\end{lemma}

\begin{proof}
   { See Appendix~\ref{sec:appendix_lemma6}.}
\end{proof}


\begin{remark}[Sensitivity of angle vs. range feedback]
\label{r:range_feedback}
As the number of antennas $M$ increases, both the required feedback bits $(p, q)$ for angle and range grow (approximately) logarithmically. However, their growth trends are different. Specifically, when $M$ is small, the beamforming performance is more sensitive to the angle errors than to the range errors (see Fig.~\ref{fig:f_AoD_r}). This indicates that when $M$ is small, more bits should be used for angle feedback (i.e., $p > q$). 
On the other hand, as $M$ becomes larger, the XL-array becomes more sensitive to range errors due to the higher pre-log term (i.e., $2\log_2 M$) in~\eqref{eq:range_bits}, thus the number of bits for range feedback $q$ should grow faster than $p$. 
Consequently, with a fixed feedback budget $B_1$ (i.e., $p+q=B_1$), the optimal bit allocation shifts from angle-dominant to range-dominant for a larger $M$.
This trend will be numerically verified in Fig.~\ref{fig:bits_allocation} of Section~\ref{sec:number_sim}.
\end{remark}
\vspace{-1.2em}
\subsection{Required Total Feedback Bits of the Three Phases}

To analyze the required feedback bits for the three-phase transmission protocol, we first characterize the upper-bound of the rate-gap $\Delta R$ between the sum-rate with perfect CSI feedback (denoted by $R_{\rm ideal}$) and the sum-rate with limited feedback (denoted by $R_{\rm lim}$). 
Specifically, following similar analysis method in \cite{alkhateebLimitedFeedbackHybrid2015}, $\Delta R$ can be upper-bounded as 
\begin{equation}
 \label{eq:rate_gap}
    \begin{aligned}
 & \Delta R \triangleq R_{\rm ideal} \!-\! R_{\rm lim}\\
 & \le \underbrace{-2\log_2 {\Gamma(\mathcal{S}_{\rm a},\mathcal{S}_{\rm r} | p,q)}}_{\triangleq \Delta R_{\rm I}} + \underbrace{\log_2\left({1}\!+\! \frac{P_{\rm tol}}{K\sigma^2} 2^{\frac{-B_2}{K-1}} \right)}_{\triangleq \Delta R_{\rm II}}.
    \end{aligned}
\end{equation}
The proof can be found in Appendix~\ref{sec:appendix_rate_gap}.  Specifically, the rate gap can be decomposed as $\Delta R=\Delta R_{\rm I}+\Delta R_{\rm II}$, where $\Delta R_{\rm I}$ comes from imperfect analog gain $\Gamma<1$ and $\Delta R_{\rm II}$ comes from RVQ with error $\mathbb{E}[\min_{\hat{\mathbf{g}}\in\mathcal{B}_2}(1-|\tilde{\mathbf{g}}^H\hat{\mathbf{g}}|^2)] = 2^{-B_2/(K-1)}$. 
To ensure a bounded rate-loss (i.e., a constant $\Delta R$) due to the feedback error, the scaling order of $B_1$ to ensure a constant $\Gamma$ is $\mathcal{O}(\log_2 M)$ (see Section \ref{sec:bits_analysis_phase1}) while that for $B_2$ to ensure a constant $\Delta R_{\rm II}$ is approximately $\mathcal{O}(K)$ (see  $\Delta R_{\rm II}$ in \eqref{eq:rate_gap}), yielding total overhead $\mathcal{O}(K+\log_2 M)$.

\section{Numerical Results}
\begin{table*}[t]
\centering
\caption{ Comparison of feedback overhead, codebook generation computational complexity, and pilot overhead.}
\label{tab:overhead_comparison}
\vspace{-0.5em}
\renewcommand{\arraystretch}{1.2}
{\begin{tabular}{lccc}
\toprule
\textbf{Scheme} & \textbf{Feedback overhead} & \textbf{Computational complexity} & \textbf{Pilot overhead} \\
\midrule
Extended geo. CB (Proposed iterative CB) & $\mathcal{O}(K(B_1+B_2))$ & $\mathcal{O}(2^{B_1}+2^{B_2})$ & $\mathcal{O}(LK) +\mathcal{O}(K^2) $ \\
Geometric CB (Proposed closed-form CB) & $\mathcal{O}(K(B_1+B_2))$ & $\mathcal{O}(2^{B_2})$ & $\mathcal{O}(LK) +\mathcal{O}(K^2)$ \\
Hyperbolic CB \cite{cuiChannelEstimationExtremely2022} & $\mathcal{O}(K(B_1+B_2))$ & $\mathcal{O}(2^{B_2})$ & $\mathcal{O}(LK)+\mathcal{O}(K^2)$ \\
Uniform CB & $\mathcal{O}(K(B_1+B_2))$ & $\mathcal{O}(2^{B_2})$ & $\mathcal{O}(LK)+\mathcal{O}(K^2)$ \\
DFT CB \cite{alkhateebLimitedFeedbackHybrid2015} & $\mathcal{O}(K(B_1+B_2))$ & $\mathcal{O}(2^{B_2})$ & $\mathcal{O}(LK)+\mathcal{O}(K^2)$ \\
Hybrid-field CB \cite{wei2022channel} & $\mathcal{O}(K(B_1+B_2))$ & $\mathcal{O}(2^{B_2})$ & $\mathcal{O}(LK)+\mathcal{O}(K^2)$ \\
\bottomrule
\end{tabular}}
\vspace{-1em}
\end{table*}
\label{sec:number_sim}
\subsection{System Setup and Benchmark Schemes}
We consider an XL-MIMO system operating at a carrier frequency of $30$ GHz, with a BS equipped with $M=387$ antennas serving $K=4$ users, if not otherwise specified.
{ The users are randomly located within the angular region $\mathcal{Q} = [-0.5,0.5]$ and the range region $\mathcal{R}=[2 D, 0.16 Z_{\rm r}]$ (i.e., around $\mathcal{R}=[4, 120]$ m).}
The SNR is set as $22$ dB \cite{ldmaorsdma}, if not otherwise specified. 
The channel gain of the LoS path is set as $\beta_{k,1}=\sqrt{\frac{\kappa}{1+\kappa}}, \forall k$ with $\kappa=9.54$ dB, and the channel gains for the NLoS paths follow $\beta_{k,\ell}\sim \mathcal{CN}(0,\sqrt{\frac{1}{1+\kappa}\frac{1}{L-1}}), \forall k=1,...,K,\forall\ell =2,...,L$~\cite{ldmaorsdma} where the number of the paths is set as $L=3$. 
{ For multi-path scenarios (i.e., non-LoS-dominant channels), the channel gains of all paths are assumed to follow $\beta_{k,\ell}\sim \mathcal{CN}(0,\frac{1}{\sqrt{L}}), \forall k,\forall\ell$. The CSI of multi-path channels is fed back via the method in \textbf{Remark~\ref{r:multi-path_channel}} with an RVQ codebook with $2^{12}$ codewords.}
If not otherwise specified, we assume that all users and scatterers are uniformly distributed in the polar sub-region designated by $\mathcal{Q}\times\mathcal{R}$, with $\times$ here denoting the Cartesian product.
 We use Monte-Carlo simulation with $1000$ runs. For performance comparison, we consider the following benchmarks. 
\begin{itemize}
    \item \textbf{Hyperbolic codebook (CB): }The polar-domain codebook proposed in \cite{cuiChannelEstimationExtremely2022} is used to feed back the locations (i.e., $\theta$ and $r$) of the users for analog beamforming, with $2^p$ angle samples and $2^q$ range samples generated by hyperbolic-range sampling in \eqref{eq:nf_r_samples_polar}.
    \item \textbf{Uniform CB: }The polar-domain codebook is used to feed back the locations of the users for analog beamforming, with $2^p$ angle samples and $2^q$ range samples generated by the uniform range sampling.
    \item \textbf{DFT CB: }The DFT codebook is used to feed back the angle of the users for analog beamforming, with $2^{p+q}$ angle samples. The codeword yielding the maximum received power is fed back for analog beamforming.
    \item \textbf{Hybrid-field CB}: The hybrid-field codebook proposed in \cite{wei2022channel} is used to feed back for analog beamforming. 
    Specifically, this hybrid-field codebook combines a far-field DFT codebook and a near-field polar-domain codebook with $2^{p}$ angle samples and $2^{q}$ range samples. The range sample is given by \eqref{eq:nf_r_samples_polar} with $\hat{r}^{\rm hf}_1=\infty, \hat{r}^{\rm hf}_{i,\forall i \ne 1}=\hat{r}^{\rm hy}_i$, where $\hat{r}^{\rm hf}_1=\infty$ is used to indicate the usage of DFT codebook.
    The codeword yielding the maximum received power is fed back.
    \item \textbf{Full CSI: }The beamforming design is based on full CSI, which serves as a rate performance upper bound.
\end{itemize}
\begin{figure}[t]
    \centering
    \begin{subfigure}{0.7\linewidth}
        \centering
        \includegraphics[width=\linewidth]{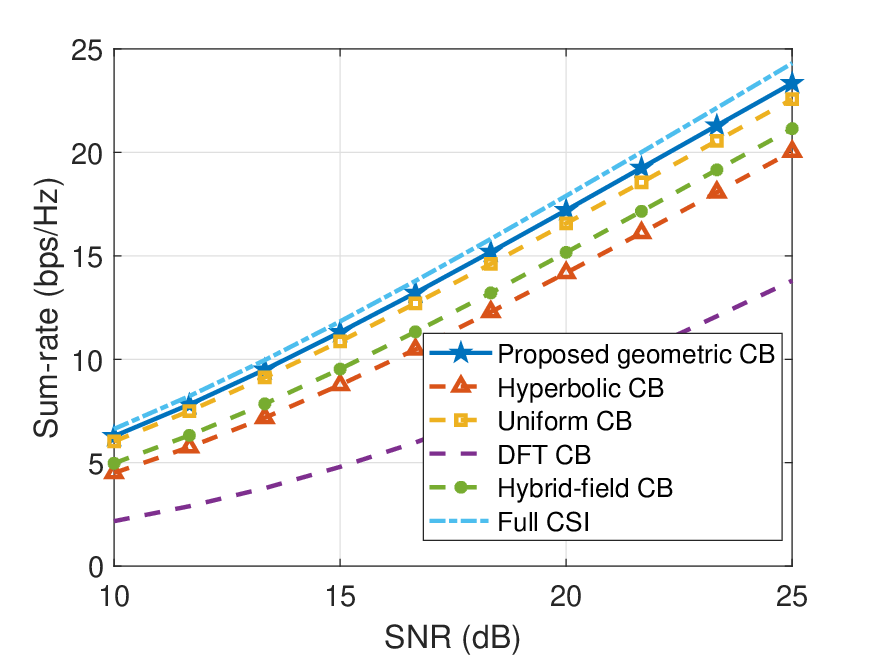}
        \subcaption{\small Uniform user range.}
    \end{subfigure}
    \begin{subfigure}{0.7\linewidth}
        \centering
        \includegraphics[width=\linewidth]{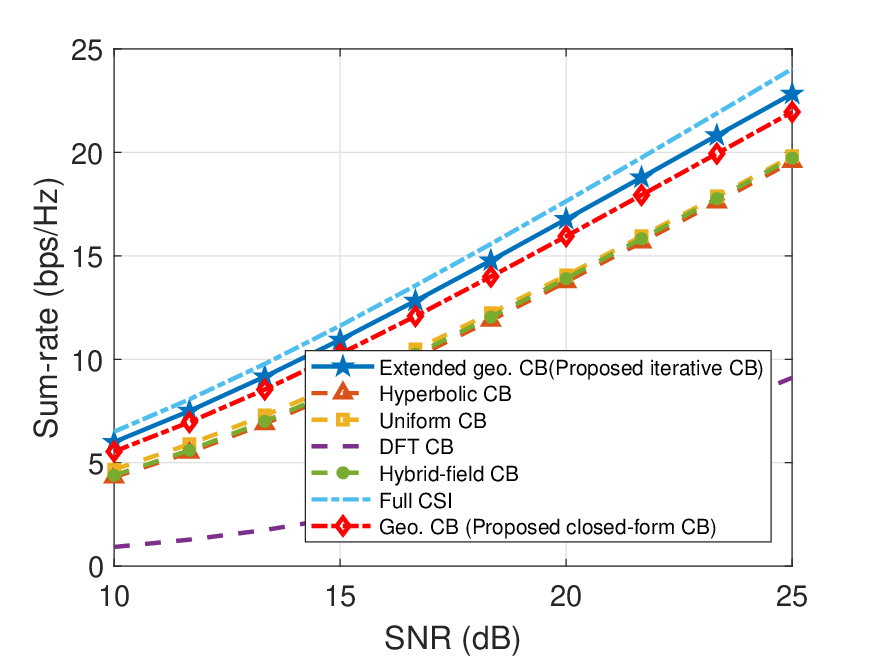}
        \subcaption{\small Hot-spot user range\footnotemark.}
    \end{subfigure}
    \vspace{-0.5em}
    \caption{\small Sum-rate vs. SNR with different user range distributions.}
    \label{fig:rate_SNR}
    \vspace{-0.8em}
\end{figure}
\footnotetext{Hot-spot distribution: user-range $r$ has a probability of $90\%$ to uniformly distribute in $[10,20]$ and $10\%$ to uniformly distribute in $[4,10]\cup [20,120]$.}
For all benchmark schemes, the angle samples are uniformly sampled over $\mathcal{Q}$. 
Unless otherwise specified, the bit allocation in Phase~1 feedback is set to $p=12$ and $q=3$. For Phase 2 feedback, we use an RVQ codebook with $B_2 = 12$ bits.

{ 
\subsection{Complexity and overhead comparison}
The overhead comparisons among different schemes are summarized in Table~\ref{tab:overhead_comparison}. 
Specifically, the feedback overhead for all limited-feedback schemes is $\mathcal{O}(K(B_1+B_2))$, where $B_1=p+q$ is the bit budget for Phase 1 feedback and $B_2$ is the bit budget for Phase 2 feedback. 
The computational complexity of the proposed iterative geometric codebook is $\mathcal{O}(2^{B_1}+2^{B_2})$ due to the need to generate both Phase 1 and Phase 2 codebooks. Our proposed closed-form geometric codebook has a lower complexity of $\mathcal{O}(2^{B_2})$ since it does not require codebook generation for Phase 1, and so do other benchmark schemes.
The pilot overhead for channel estimation in all limited-feedback schemes is $\mathcal{O}(LK)+\mathcal{O}(K^2)$, where $\mathcal{O}(LK)$ is for estimating the $L$-path channel of each user and $\mathcal{O}(K^2)$ is for effective channel estimation in Phase 2 (the size of effective channels are $K\times K$). 
}

\vspace{-1em}
\subsection{Performance of the Proposed Sampling Design}
To evaluate the effectiveness of the proposed sampling design for the limited-feedback codebook, we show the achievable sum-rate versus the SNR in Fig.~\ref{fig:rate_SNR} under different user range distributions. 
Specifically, the proposed closed-form geometric codebook is applied to the uniform user distribution case [see Fig.~\ref{fig:rate_SNR}(a)], while the extended geometric codebook is employed for the non-uniform user distribution case [see Fig.~\ref{fig:rate_SNR}(b)]. 
First, it is observed that the proposed geometric codebook consistently achieves the highest rate performance, approaching that of full CSI, and significantly outperforming the DFT codebook. 
This advantage arises from its efficient range feedback, which enables the BS to perform more effective analog beamforming, while the DFT codebook lacks range information.
Second, compared with other range-feedback benchmark schemes (i.e., hyperbolic and hybrid-field codebooks), the proposed geometric codebook achieves better rate performance.
This is because, the hyperbolic codebook allocates too few samples in the long range regime, resulting in inadequate user coverage, while the hybrid-field codebook, although benefiting from the inclusion of the DFT codebook for far users, still fails to adapt to the underlying user distribution.
Moreover, it is observed that the uniform codebook achieves rate performance comparable to that of the proposed geometric codebook under uniform user distribution [Fig.~\ref{fig:rate_SNR}(a)]. 
This is because the uniform range samples can cover far-away users, incurring only a slight rate loss due to the neglect of near-field sampling density considerations.
However, under non-uniform user distributions [Fig.~\ref{fig:rate_SNR}(b)], the performance gain of the proposed codebook over the uniform one becomes more pronounced, as the uniform design cannot adapt to uneven user locations, resulting in significant mismatch and rate degradation.
It is worth noting that although the closed-form geometric codebook is derived under the assumption of uniform user distribution, it demonstrates strong robustness, providing satisfactory sum-rate performance even under non-uniform distributions and offering an attractive trade-off between performance and complexity.

\begin{figure}[t]
    \centering
    \includegraphics[width=0.65\linewidth]{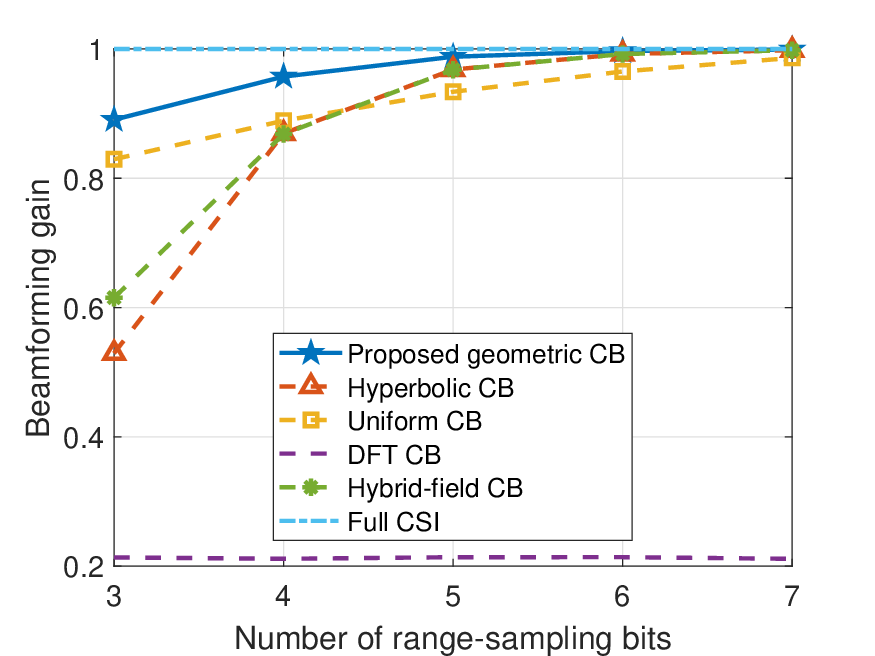}
    \caption{\small Beamforming gain vs. number of range-sampling bits $q$.}
    \label{fig:rate_loss_p}
    \vspace{-1.2em}
\end{figure}
\begin{figure}[t]
    \centering
    \includegraphics[width=0.65\linewidth]{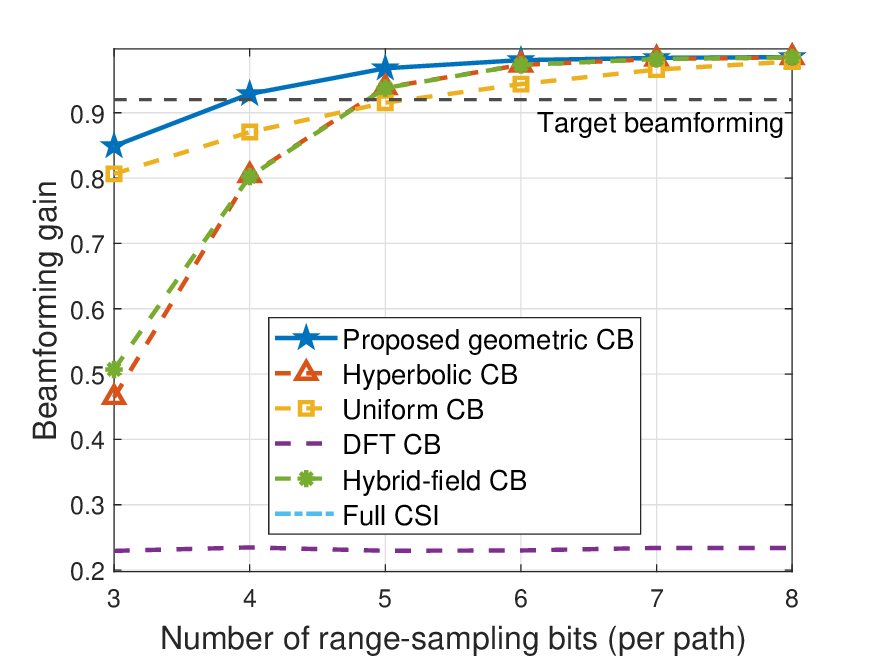}
    \caption{\small Beamforming gain vs. number of range-sampling bits in multi-path scenario.}
    \label{fig:rate_L}
    \vspace{-1.2em}
\end{figure}

\begin{figure}[t]
    \centering
    \includegraphics[width=0.65\linewidth]{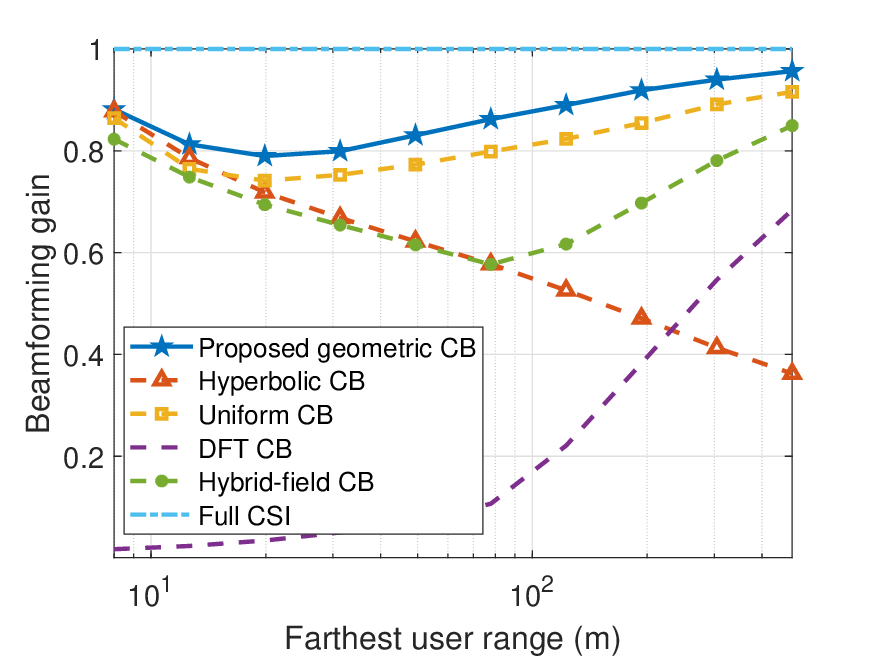}
    \caption{\small Beamforming gain vs. farthest user range $r_{\max}$.}
    \label{fig:f_bits_codebook_rmax}
    \vspace{-1.2em}
\end{figure}

In Fig.~\ref{fig:rate_loss_p}, we show the beamforming gain of different codebooks vs. the number of range-sampling bits $q$. 
It is observed that, the superiority of the geometric codebook becomes particularly evident when the number of range-sampling bits is small. 
The performance gain of the proposed geometric codebook becomes increasingly pronounced in the low range-sampling bit regime (e.g., $q =3$ or $q=4$).
This is because the geometric codebook effectively balances the range-sampling density across the near and long range regions.
Moreover, it is observed that, when the number of range-sampling bits is sufficiently large (e.g., $q \ge 5$), both the hyperbolic and hybrid-field codebooks achieve better beamforming performance than that of the uniform codebook. 
This is attributed to the inability of uniform codebook to adapt to \emph{near-field} effect, where denser sampling is required for near-range users, while the hyperbolic and hybrid-field codebooks alleviate this limitation through their non-uniform sampling strategies.

{ Moreover, to evaluate the proposed geometric codebook in multi-path case, in Fig.~\ref{fig:rate_L}, we plot the beamforming gain vs. the number of range-sampling bits $q$ (per path). It is observed that, the proposed geometric codebook still outperforms other benchmark schemes in multi-path scenarios, demonstrating its effectiveness beyond LoS-dominant channels. For example, when the target beamforming gain is $\Gamma_0=0.92$, the proposed geometric codebook can save 3 bits (accounting for all $L=3$ paths) of range-sampling feedback overhead compared to the hyperbolic codebook and the hybrid-field codebook, and 6 bits compared to the uniform codebook.
}


In Fig.~\ref{fig:f_bits_codebook_rmax}, we plot the beamforming gain $\Gamma$ vs. the farthest user range $r_{\max}$. Specifically, for different $r_{\max}$, we re-generate the range samples of all codebooks according to their respective sampling strategies with the fixed number of range sampling bits, while keeping the angle samples unchanged.
Some interesting observations can be made.
First, the beamforming performance of the hyperbolic and DFT codebooks exhibits opposite trends as $r_{\max}$ increases. 
Specifically, the beamforming performance of the hyperbolic codebook deteriorates with increasing $r_{\max}$, since it allocates an excessive number of samples to the near-range region and fails to match the uniform user range distribution, resulting in severe range-sampling errors for distant users, thereby degrading the beamforming gain as $r_{\max}$ increases.
In contrast, the beamforming performance of the DFT codebook improves as $r_{\max}$ increases. This is because for a larger $r_{\max}$, the users are more likely to be located in the far-field region, making the far-field DFT codebook sufficient for beamforming.
Therefore, by combining hyperbolic and DFT codebooks, the beamforming gain of the hybrid-field codebook first decreases (due to the influence of the hyperbolic component) and then increases (due to the contribution of the DFT component) as $r_{\max}$ increases.
{ The geometric and uniform codebooks exhibit similar non-monotonic trends; whose beamforming gains first decrease and then increase with an increasing $r_{\max}$. 
This is because, in the small-$r_{\max}$ regime, near-field effects dominate, making accurate range sampling more critical. In this regime, increasing $r_{\max}$ leads to a wider range region (given a fixed number of range sample bits), thereby increasing the range-sampling errors for these two codebooks and suffering reduced beamforming gains.}
As $r_{\max}$ continues to increase and users are predominantly located in the far-field region, the near-field effects start to diminish, leading to improved beamforming gains.
Across all regimes, the geometric codebook consistently outperforms the uniform counterpart, benefiting from its more efficient range-sampling strategy that captures channel variations over distance.

\vspace{-1em}
\subsection{Performance of Bit Allocation}

\begin{figure}[t]
    \centering
    \includegraphics[width=0.65\linewidth]{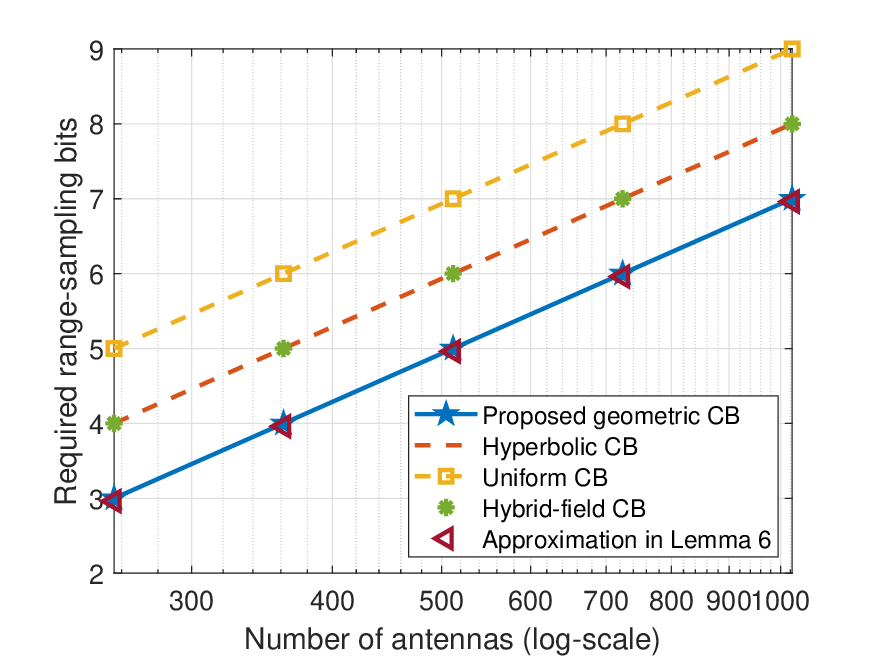}
    \caption{\small Required $q$ vs. number of antennas $M$.}
    \label{fig:bits_required}
    \vspace{-1.2em}
\end{figure}

In this subsection, we evaluate the impact of the number of antennas $M$ on the proposed codebook design and the corresponding bit allocation strategy. Specifically, it is worth noting that, full CSI scheme requires $B_1=\infty$ bits. Moreover, the DFT codebook can never achieve a given target beamforming gain of $\Gamma_0 = 0.95$ due to the lack of range feedback. As such, we omit these two schemes in the following simulations.

In Fig.~\ref{fig:bits_required}, we show the required range-sampling bits to achieve a target beamforming gain of $\Gamma_0 = 0.95$ vs. the number of antennas $M$. 
It is observed that, all codebooks require a higher overhead to achieve a targeted beamforming gain as $M$ increases. This is because a larger $M$ leads to finer spatial resolution, thus demanding more accurate range feedback.
Owning to its efficient range sampling strategy, the proposed geometric codebook achieves the lowest overhead, reducing up to $2$ bits per user per path compared to the uniform codebook.
Furthermore, it is observed that the required range-sampling bits scale as $\mathcal{O}(\log_2 M)$, which is consistent with the analysis in \textbf{Lemma~\ref{c:epc_p}}. 

In Fig.~\ref{fig:bits_allocation}, we show the expected beamforming gain vs. the number of antennas $M$. The total bit overhead is fixed at $p+q=16$.
It is observed that, the beamforming gain of all schemes decreases as $M$ increases. 
This is because larger antenna arrays require higher feedback overhead to preserve the same beamforming resolution, thus leading to inevitable performance degradation due to the fixed number of angle and range sample bits, which is consistent with the observation in Fig.~\ref{fig:bits_allocation}.
Moreover, the performance gap between the proposed geometric codebook and other benchmarks becomes larger as $M$ increases. 
This is due to the fact that, as $M$ grows, near-field effects become more pronounced (i.e., larger Rayleigh distance), making efficient range sampling critical for achieving high beamforming gain, thus demonstrating the effectiveness of the proposed geometric codebook in covering users under practical user distribution.

In Fig.~\ref{fig:bits_allocation_pair}, we show the optimal bit allocation $(p,q)$ vs. the number of antennas $M$ for the proposed geometric codebook, given a fixed total bit budget of $p+q=16$.
It is observed that, the optimal number of range-sampling bits $q$ increases while the optimal number of angle-sampling bits $p$ decreases as $M$ increases, which agrees with the analysis in \textbf{Remark~\ref{r:range_feedback}}.
\begin{figure}[t]
    \centering
    \begin{minipage}{0.45\linewidth}
        \centering
        \includegraphics[width=0.75\linewidth]{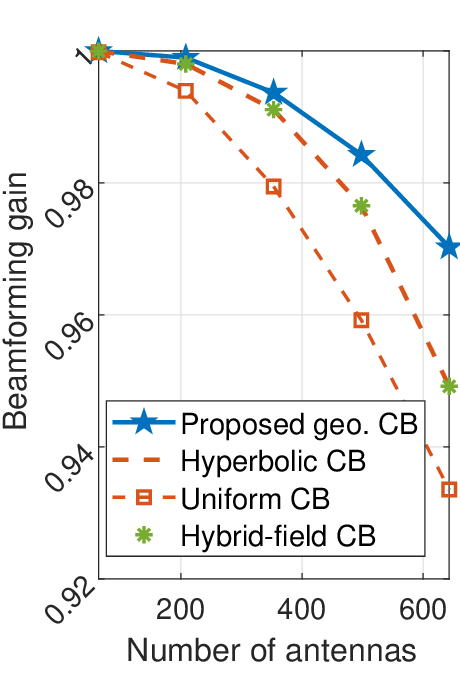}
        \caption{\small Beamforming gain vs. number of antennas $M$.}
        \label{fig:bits_allocation}
    \end{minipage}
    \begin{minipage}{0.45\linewidth}
        \centering
        \includegraphics[width=0.75\linewidth]{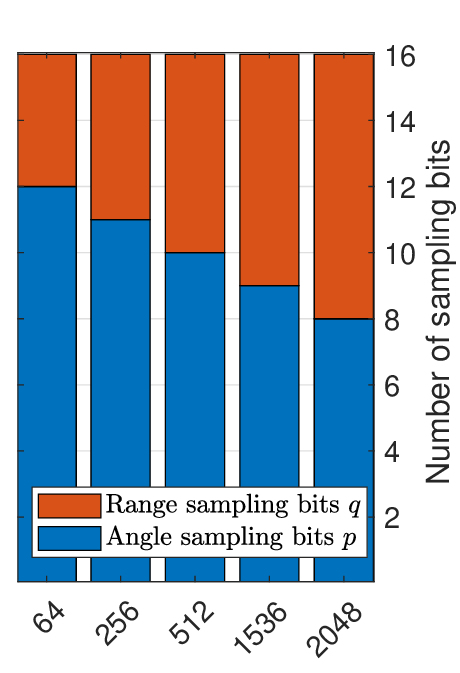}
        \caption{\small Optimal bit allocation vs. number of antennas  $M$.}
        \label{fig:bits_allocation_pair}
    \end{minipage}
    \vspace{-1.2em}
\end{figure}
\section{Conclusion}
\label{sec:conclusion}
In this paper, we proposed a novel limited-feedback codebook for FDD XL-MIMO systems, addressing the excessive overhead of existing designs by exploiting practical user distributions. For uniformly distributed users, we showed that uniform angle sampling and geometric range sampling provide near-optimal performance, while this design framework was extended to non-uniform distributions via an alternating sampling method. Theoretical analysis revealed that, as the array size increases, the optimal feedback allocation favors range samples over angle samples. Numerical results demonstrated that the proposed codebook consistently achieves superior rate performance over benchmark schemes and robustness across various system configurations. Overall, the proposed design provides an efficient and practical CSI feedback solution for near-field XL-MIMO systems, balancing beamforming performance and feedback overhead.

\appendices
\section{}
\label{sec:appendix_bfgain}
Given sufficiently small angle and range sampling errors, we can approximate the analog beamforming gain as follows
\begin{equation}
    \begin{split}
     \label{eq:beamforming_gain_appendix}
     & \left|\bfa^H(\theta,r) \bfa(\hat{\theta},\hat{r})\right|  \\
     & \overset{(a)}{\approx} \frac{2}{M} \left| \sum_{m=0}^{(M-1)/2} e^{\jmath m \pi |\theta-\hat{\theta}|} e^{\jmath k_c m^2 d_0^2 \left|\frac{1-\theta^2}{2r} - \frac{1-\hat{\theta}^2}{2 \hat{r}}\right|} \right|\\
     &\overset{(b)}{\approx}\frac{2}{M} \left| \sum_{m=0}^{(M-1)/2} e^{\jmath m \pi |\theta-\hat{\theta}|} e^{\jmath k_c m^2 d_0^2\frac{1-{\theta}^2}{2} \left|\frac{1}{r}-\frac{1}{\hat{r}}\right|} \right|  \triangleq f(\varepsilon_\theta,\varepsilon_{r}), 
\end{split} 
\end{equation}
where $\varepsilon_\theta \triangleq |\theta - \hat{\theta}|$ and $\varepsilon_{r} \triangleq (1-\theta^2) \left| \frac{1}{r} - \frac{1}{\hat{r}} \right|$ are the angle and angle-range sampling errors, respectively. In this equation, the approximation $(a)$ holds due to the approximate symmetry of the array response when $\varepsilon_\theta$ and $\varepsilon_{r}$ are small~\cite{DFTcodebookXLbeamtraining} and the Fresnel approximation~\cite{lywTut}. Moreover, approximation $(b)$ is due to the truth that $\frac{1-\theta^2}{2} \approx \frac{1-\hat{\theta}^2}{2}$ when $p$ is sufficiently large. 
{ Specifically, we plot the approximation of overset (a) and (b) of  in Fig.~\ref{fig:approximation_valid} to verify their accuracy. It is observed that both approximations are accurate when the angle sampling error is sufficiently small (e.g., set $q\ge 12$ bits).}

\section{}
\label{sec:appendix_expect_decouple}
For ease of exposition, we define $\varepsilon_{\rm u}\triangleq \vartheta \varepsilon_{r}$ (and its tilde version $\tilde{\varepsilon}_{\rm u}\triangleq \vartheta \tilde{\varepsilon}_{r}$).
The maximum of the codeword selection problem in~\eqref{eq:beamforming_gain} can be approximately reformulated as
\begin{equation}
 \label{eq:beamforming_gain_appendix_B}
    \begin{aligned}
        \Gamma  = \mathbb{E}_{(\theta,r)} \left[\max_{\hat{\theta} \in \Scal_{\rm a}, \hat{r} \in \Scal_{\rm r}} f_\theta({\varepsilon}_\theta,{\varepsilon}_{\rm u})\right] \approx \mathbb{E}_{(\theta,r)}[f(\tilde{\varepsilon}_{\theta},\tilde{\varepsilon}_{\rm u})].
    \end{aligned}
\end{equation}
Specifically, for ease of exposition, by defining $X\triangleq \tilde{\varepsilon}_\theta$ and $Y\triangleq \tilde{\varepsilon}_{\rm u}$, the expectation in~\eqref{eq:beamforming_gain_appendix_B} can be approximated as follows~\cite{casella2002statistical}
\begin{equation*}
    \mathbb E_{(\theta,r)}[f]\approx f(\mu_X,\mu_Y)+
    \frac12\!\left(
 f_{xx}\sigma_X^2
 +2f_{xy}\sigma_{XY}
 +f_{yy}\sigma_Y^2
 \right)
\end{equation*}
with $\mu_X = \mathbb{E}_\theta[\tilde{\varepsilon}_\theta]$, $\mu_Y = \mathbb{E}_{(\theta,r)}[\tilde{\varepsilon}_{\rm u}]$, $\sigma_X^2 = \mathrm{Var}[\tilde{\varepsilon}_\theta]$, $\sigma_Y^2 = \mathrm{Var}[\tilde{\varepsilon}_{\rm u}]$, and $\sigma_{XY} = \mathrm{Cov}(\tilde{\varepsilon}_\theta,\tilde{\varepsilon}_{r})$, where $f_{xx}$, $f_{xy}$, and $f_{yy}$ are the second-order partial derivatives of $f$ w.r.t. $\tilde{\varepsilon}_\theta$ and $\tilde{\varepsilon}_{\rm u}$. For ease of exposition, we define $\varpi\triangleq \max_{\theta} \tilde{\varepsilon}_\theta - \min_{\theta} \tilde{\varepsilon}_\theta$ and $\tilde{\varpi}\triangleq \max_{r} \tilde{\varepsilon}_{r} - \min_{r} \tilde{\varepsilon}_{r}$.
Based on the Popoviciu's inequality on variances \cite{casella2002statistical},
$\sigma_{X}^2$ and $\sigma_{Y}^2$ can be bounded as $\sigma_X^2 = \mathrm{Var}[\tilde{\varepsilon}_\theta] \leq \frac{\varpi^2}{4}$ and $ \sigma_Y^2 = \mathrm{Var}[\tilde{\varepsilon}_{\rm u}] = \mathrm{Var}[\vartheta \tilde{\varepsilon}_r] \leq \mathrm{Var}[\vartheta]\frac{\tilde{\varpi}^2}{4}$.
Moreover, based on the Cauchy-Schwarz inequality \cite{cvxOpt}, the covariance $\sigma_{XY}$ is bounded as $\sigma_{XY} \leq \sqrt{\mathrm{Var}[\tilde{\varepsilon}_\theta] \mathrm{Var}[\tilde{\varepsilon}_{\rm u}]}\leq \mathrm{Var}[\vartheta] \frac{\varpi\tilde{\varpi}}{4}$.
Thus, when $\varpi$ (and $\tilde{\varpi}$) tends to zero (i.e., the sampling error is sufficiently small), we have $\sigma_X^2, \sigma_Y^2$ and $\sigma_{X,Y}$ tending to zero with order $\mathcal{O}(\varpi^2)$, $\mathcal{O}(\tilde{\varpi}^2)$ and $\mathcal{O}(\varpi\tilde{\varpi})$, while $f(\mu_X,\mu_y)\gg 0$, which directly implies the result.
\vspace{-1em}
\section{}\begin{figure}[t]
    \centering
    \includegraphics[width=\linewidth]{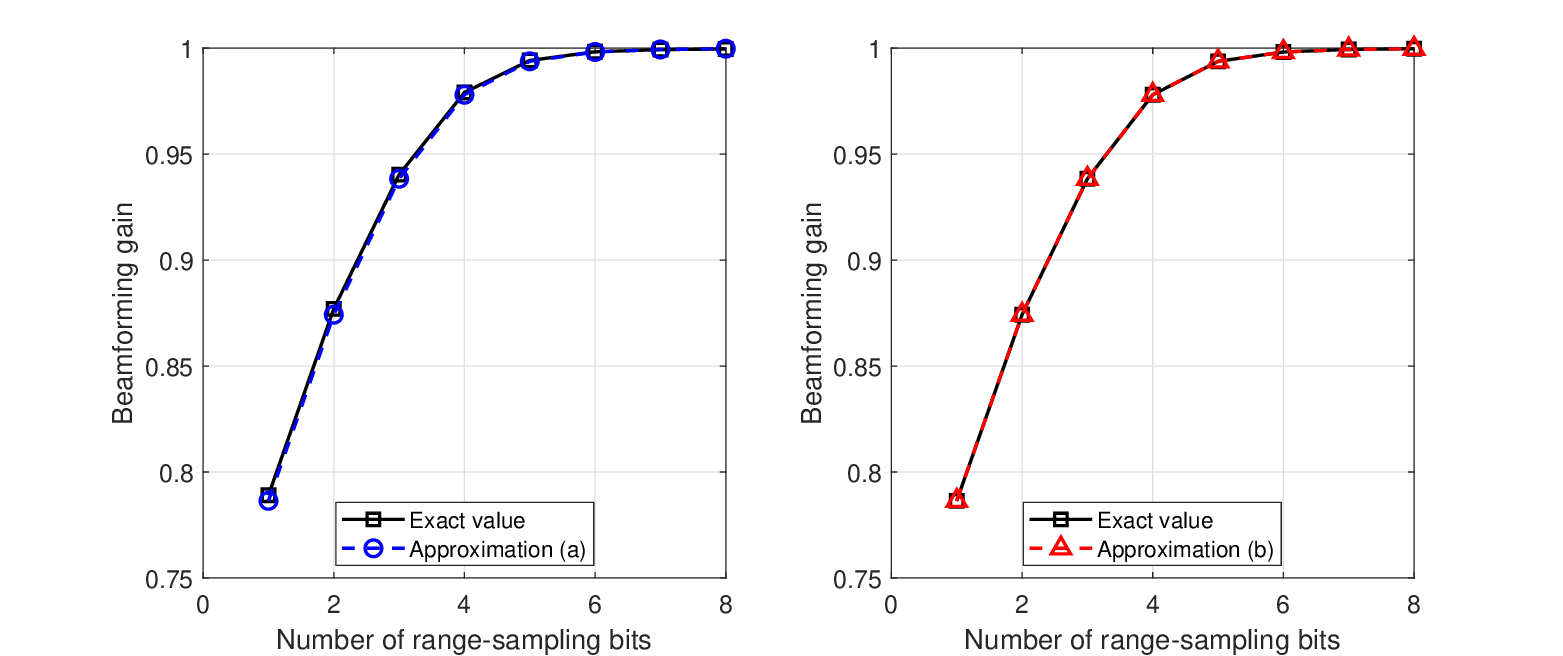}
    \caption{\small Verification of approximations in~\eqref{eq:beamforming_gain_appendix} with $M=387$. For other parameters, please refer to Section~\ref{sec:number_sim}.}
    \label{fig:approximation_valid}
    \vspace{-1.em}
\end{figure}
\label{sec:appendix_angle_sampling}
For a given $\mathcal{C}_i$, $i=1,2,\dots,2^p$, we have $\mathrm{Pr}(\theta\in \mathcal{C}_i)={D(\mathcal{C}_i)}/{D(\mathcal{Q})}$. 
The objective function of Problem (P8) is equivalent to $\expect_\theta[\tilde{\varepsilon}_{\theta}] = \sum_{i=1}^{2^p} \expect_\theta[\tilde{\varepsilon}_{\theta}|\mathcal{C}_i] \mathrm{Pr}(\theta \in \mathcal{C}_i) = \sum_{i=1}^{2^p} \frac{(a_i^2+a_{i+1}^2)}{D(\mathcal{Q})}$ by substituting $\mathrm{Pr}(\theta\in \mathcal{C}_i)$ into~\eqref{eq:Voronoi_cell_angle_sampling}
where $a_i \triangleq (\acute{\theta}_i -\acute{\theta}_{i-1})$. We aim at minimizing $\mathbb{E}_\theta[\tilde{\varepsilon}_\theta]=\sum_{i=1}^{2^p} {(a_i^2+a_{i+1}^2)}/{D(\mathcal{Q})}$ subject to $\sum_{i=1}^{2^p+1} a_i = D(\mathcal{Q})$.
It can be readily verified that $a_i=a_j,\forall i \ne j$ is the optimal solution to the above problem~\cite{llyodLeastSquaresQuantization1982}. Thus, we have $a_i = a_j = \frac{D(\mathcal{Q})}{2^p+1}$, which directly leads to the optimal $\hat{\theta}_i = \frac{a_i + a_{i+1}}{2} = \frac{D(\mathcal{Q})}{2(2^p+1)}$ for $i=1,2,\dots,2^p$.
\vspace{-1em}
\section{}
\label{sec:appendix_range_sampling_cells}
Let $r\sim\mathcal{U}(\dot r_i,\dot r_{i+1})$, define $\varsigma=1/r$.
Then $\mu_\varsigma(\varsigma)=\frac{1}{\varsigma^2 D(\mathcal{I}_i)}$, $\varsigma\in[\dot\varsigma_i,\dot\varsigma_{i+1}]$ with
$\dot\varsigma_i=1/\dot r_{i+1}<\dot\varsigma_{i+1}=1/\dot r_i$.
Let $\hat\varsigma=1/\hat r$. The conditional expected (inverse-range) error is
$\mathbb{E}_r[\tilde{\varepsilon}_{r}|\mathcal{I}_i]
=\int_{\dot\varsigma_i}^{\dot\varsigma_{i+1}} |\varsigma-\hat\varsigma| \mu_\varsigma(\varsigma)\,d\varsigma
= \frac{-2+\hat\varsigma(1/\dot\varsigma_i+1/\dot\varsigma_{i+1})
-2\log\hat\varsigma-\log(\dot\varsigma_i\dot\varsigma_{i+1})}{D(\mathcal{I}_i)}.$
With the first-order optimality condition, we obtain
$\partial / \partial \hat\varsigma =0 \;\Rightarrow\;
\hat\varsigma^*=\frac{2}{1/\dot\varsigma_i+1/\dot\varsigma_{i+1}}
\;\Rightarrow\;
\hat r_i^*=\frac{\dot r_i+\dot r_{i+1}}{2}$, which directly leads to the optimal $\hat{r}_i = \frac{r_i + r_{i+1}}{2}$ for $i=0,1,\dots,2^q$.
\vspace{-1em}
\section{}
\label{sec:appendix_range_sampling} 
By substituting the obtained $\mathbb{E}_r[\tilde{\varepsilon}_{r}|\mathcal{I}_i]$ and $\mathrm{Pr}(r \in \mathcal{I}_i)$ into \eqref{eq:problem_r_approx_k} and ignoring constants, Problem (P10) is equivalent to minimizing the objective $E \triangleq \sum_{i=0}^{2^{q}} \log \left(\xi_i + \frac{1}{\xi_i} + 2\right)$, where $\xi_i=\frac{\dot r_{i+1}}{\dot r_i}$ subject to $\prod_{i=1}^{2^q} \xi_i = \frac{r_{\mathrm{max}}}{r_{\mathrm{min}}}$. Given the equality $\sum_{i=1}^{2^q}\log \left(\xi_i + \frac{1}{\xi_i} + 2\right)= \log \prod_{i=1}^{2^q}\left(\xi_i+\frac{1}{\xi_i}+2\right)$, it is easy to verify that $\prod_{i=1}^{2^q}\left(\xi_i+\frac{1}{\xi_i}+2\right)$ is minimized when $\xi_i=\xi_j$ for all $1\le i\ne j\le 2^q$~\cite{cvxOpt}. Let $\xi_i=\xi$ for all $i$, then $\prod_{i=1}^{2^q}\left(\xi_i+\frac{1}{\xi_i}+2\right)=\left(\xi+\frac{1}{\xi}+2\right)^{2^q}$. Thus, the expected range sampling error is given by $\mathbb{E}_r[\tilde{\varepsilon}_{r}] = \frac{2^{q+1}(1-\theta^2)}{D(\mathcal{R})}\left(\log \frac{\sqrt[2^{p+1}]{\xi}+1/\sqrt[2^{p+1}]{\xi}}{2}\right)$, where $\xi=\frac{r_{\mathrm{max}}}{r_{\mathrm{min}}}$.
\vspace{-1em}

{\section{}
\label{sec:appendix_lemma5}
When $\tilde{\varepsilon}_\theta$ is sufficiently small and range sampling is accurate (i.e., $\mathbb{E}_r[\tilde{\varepsilon}_{r}] \approx 0$), the expected beamforming gain in~\eqref{eq:approximation_M_E} can be approximated as
\begin{equation*}
    \Gamma \approx \frac{2}{M} \left| \sum_{m=0}^{(M-1)/2} e^{\jmath m \pi \mathbb{E}_\theta[\tilde{\varepsilon}_\theta]} \right|.
\end{equation*}
By applying the geometric series summation formula and following similar analysis as in~\cite{DFTcodebookXLbeamtraining}, this can be further simplified to
\begin{equation*}
    \Gamma \approx \left|\frac{\sin(\frac{1}{2}M \pi \mathbb{E}_\theta[\tilde{\varepsilon}_\theta])}{\frac{1}{2}M\pi \mathbb{E}_\theta[\tilde{\varepsilon}_\theta]}\right|.
\end{equation*}
To achieve the target beamforming gain $\Gamma = \Gamma_0$, we require
\begin{equation*}
    \left|\frac{\sin(\frac{1}{2}M \pi \mathbb{E}_\theta[\tilde{\varepsilon}_\theta])}{\frac{1}{2}M\pi \mathbb{E}_\theta[\tilde{\varepsilon}_\theta]}\right| = \left|\frac{\sin(\frac{1}{2}M_0 \pi \tilde{\varepsilon}_{\theta,\Gamma_0})}{\frac{1}{2}M_0\pi \tilde{\varepsilon}_{\theta,\Gamma_0}}\right|,
\end{equation*}
where the right-hand side is a constant determined by the target gain $\Gamma_0$ and the reference array size $M_0$. When the argument is sufficiently small (i.e., $\frac{1}{2}M \pi \mathbb{E}_\theta[\tilde{\varepsilon}_\theta] \ll 1$), the sinc function is approximately linear, yielding
\begin{equation*}
    \mathbb{E}_\theta[\tilde{\varepsilon}_\theta] M \approx \tilde{\varepsilon}_{\theta,\Gamma_0} M_0.
\end{equation*}
From the optimal angle sampling in~\eqref{eq:nf_aod_samples}, we have $\mathbb{E}_\theta[\tilde{\varepsilon}_\theta] = \frac{D(\mathcal{Q})}{4 \cdot 2^{p}}$. Substituting this into the above equation and solving for $p$, we obtain the desired result.
\vspace{-1em}
\section{}
\label{sec:appendix_lemma6}
We first derive the asymptotic behavior of the expected range sampling error. Based on the Taylor expansion of the function $x \log_2 \left(\frac{\sqrt[x]{\xi}+1/\sqrt[x]{\xi}}{2}\right)$ in~\eqref{eq:optimal_range_error_sum} around large $x$ (i.e., $x = 2^q$ with sufficiently large $q$), we obtain
\begin{equation*}
    \begin{aligned}
        &x \log_2 \left(\frac{\sqrt[x]{\xi}+1/\sqrt[x]{\xi}}{2}\right) 
        = x \log_2 \left(\frac{e^{\frac{\log \xi}{x}} + e^{-\frac{\log \xi}{x}}}{2}\right) \\
        &= x \log_2 \left(\cosh\left(\frac{\log \xi}{x}\right)\right) 
        \approx x \log_2 \left(1 + \frac{(\log \xi)^2}{2x^2}\right) \\
        &\approx \frac{(\log \xi)^2}{2x \ln 2},
    \end{aligned}
\end{equation*}
where the first approximation uses the Taylor expansion $\cosh(y) \approx 1 + \frac{y^2}{2}$ for small $y$, and the second approximation uses $\log_2(1+z) \approx \frac{z}{\ln 2}$ for small $z$. Substituting $x = 2^q$ into the above result and recalling~\eqref{eq:optimal_range_error_sum}, we have
\begin{equation*}
    \begin{aligned}
        \expect_r[\mathring{{\varepsilon}}_{r}] = \frac{2^{q}}{D(\mathcal{R})}\left(\log \frac{\sqrt[2^{q}]{\xi}+1/\sqrt[2^{q}]{\xi}}{2}\right) \approx \frac{\log^2 \xi}{D(\mathcal{R}) 2^{q+1} \ln 2}.
    \end{aligned}
\end{equation*}

Next, we analyze how the target beamforming gain $\Gamma_0$ constrains the range sampling error. Following the approximation framework in~\eqref{eq:approximation_M_E} and leveraging the near-field beamforming gain analysis in~\cite[\textbf{Lemma 1}]{cuiChannelEstimationExtremely2022}, when the angle sampling is sufficiently accurate (i.e., $\mathbb{E}_\theta[\tilde{\varepsilon}_\theta] \approx 0$), to the target beamforming gain $\Gamma = \Gamma_0$, we require
\begin{equation*}
    M^2 \mathbb{E}_r[\tilde{\varepsilon}_{r}] \approx \check{M}_0^2 \tilde{\varepsilon}_{r,\Gamma_0},
\end{equation*}
where $\check{M}_0$, $\tilde{\varepsilon}_{r,\Gamma_0}$ are constants satisfying $\check{f}(0,\check{\vartheta} \tilde{\varepsilon}_{r,\Gamma_0},\check{M}_0) = \Gamma_0$. After substituting the asymptotic expression of $\mathbb{E}_r[\tilde{\varepsilon}_{r}]$ into the above equation and solving for $q$, we obtain the desired result.}
\vspace{-1em}

\section{}
\label{sec:appendix_rate_gap}
Followed by \cite{alkhateebLimitedFeedbackHybrid2015}, the rate gap $\Delta R$ is upper-bounded by
\begin{align}
\Delta R &\le \expect\lmidbra \log_2 \!\left(\frac{|\mathbf{a}^H(\theta,r)\mathbf{a}(\theta,r)|^2}{|\mathbf{a}^H(\theta,r)\mathbf{a}(\hat{\theta},\hat{r})|^2 }\right) \rmidbra \notag\\
&\quad + \expect \lmidbra \log_2 \!\left( 1+\frac{P_{\rm tol} (K-1)}{\sigma^2 K} \|\mathbf{g}\|^2 |\hat{\mathbf{g}}^H \tilde{\mathbf{f}}_{\rm BB}|^2  \right) \rmidbra
\label{eq:rate_gap_appendix}
\end{align}
where $\tilde{\mathbf{f}}_{\rm BB}$ is the digital beamforming vector with perfect CSI. 
Herein, $\mathbf{g}$ (and $\hat{\mathbf{g}}$ representing the corresponding codeword in Phase 2) and $\tilde{\mathbf{f}}_{\rm BB}$ are the effective channel and digital beamforming of \emph{different users}, respectively. Specifically, utilizing the convexity of $\log$ function, the first term in the above equation can be approximated as
\begin{equation*}
    \begin{aligned}
 &\expect\lmidbra \log_2 \(1/(\mathbf{a}^H(\theta,r) \mathbf{a}(\hat{\theta},\hat{r}))^2\) \rmidbra\\
    \le &-\log_2 \left(\mathbb{E}\left[(\mathbf{a}^H(\theta,r) \mathbf{a}(\hat{\theta},\hat{r}))^2\right]\right)\!\! \approx \!\! -\!\log_2\!\! \left(\mathbb{E}\left[f^2(\tilde{\varepsilon}_\theta,\tilde{\varepsilon}_r)\right]\right)\\
 = & -\log_2 \left(\mathbb{E}^2[f]+\operatorname*{Var}[f]\right) \overset{(c)}{\approx}-\log_2 \left(\mathbb{E}^2[f]\right)= -2\log_2 \Gamma,
    \end{aligned}
\end{equation*}
where $(c)$ is due to that $\operatorname{Var}[f]\approx 0$. Moreover, we can upper-bound the second term in \eqref{eq:rate_gap_appendix} as
\begin{equation*}
    \begin{aligned}
 &\expect \lmidbra \log_2 \( 1\!+\!\frac{P_{\rm tol} (K-1)}{\sigma^2 K} \|\mathbf{g}\|^2 |\hat{\mathbf{g}}^H \tilde{\mathbf{f}}_{\rm BB}|^2  \) \rmidbra \le \\
 &\log_2 \( 1\!+\!\frac{P_{\rm tol} (K-1)}{\sigma^2 K} \expect[\|\mathbf{g}\|^2] \expect [|\hat{\mathbf{g}}^H \tilde{\mathbf{f}}_{\rm BB}|^2]  \) \overset{(d)}{\le} \\
 &\log_2 \( 1+\frac{P_{\rm tol}}{\sigma^2 K} 2^{\frac{-B_2}{K-1}} \),
    \end{aligned}
\end{equation*}
where the bound $(d)$ can be obtained by similar tricks in~\cite[\textbf{Theorem~7}]{alkhateebLimitedFeedbackHybrid2015}. Thus we complete the proof.
\bibliographystyle{IEEEtran}
\bibliography{bib_lib.bib}

\end{document}

%% file: macro.tex
\newcommand{\thd}[1]{$#1$-th}
\newcommand{\dt}[1]{d_{\bm{t};{{#1}}}}
\newcommand{\du}[1]{d_{{{#1}};\bm{r}}}
\newcommand{\wn}{\frac{2\pi}{\lambda}}
\newcommand{\expect}{\operatorname{\mathbb{E}}}
\newcommand{\ch}{\mathbf{h}}
\newcommand{\Cq}{ C_\theta }
\newcommand{\Cr}{C_r }
\newcommand{\diag}{\operatorname{diag}}
\newcommand{\deltaT}{\Delta \hat{\theta}}
\newcommand{\epsT}{\varepsilon_\theta}
\newcommand{\epsL}{\varepsilon_r}
\newcommand{\snr}{{\rm SNR}}
\newcommand{\cod}{\Delta C(\hatT,\theta,\hatD,r)}
\newcommand{\Capp}{{{C}}}
\newcommand{\cae}{{\mathcal{E}}}
\newcommand{\bfa}{{\mathbf{a}}}
\newcommand{\codew}{{\boldsymbol{w}}}
\newcommand{\rrg}{r_{\rm rg}}
\newcommand{\lrg}{\ell_{\rm rg}}
\newcommand{\minimize}{\operatorname*{minimize}}
\newcommand{\trg}{\theta_{\rm rg}}
\newcommand{\rrgq}{\rho_{\rm rg}}
\newcommand{\cFullFun}{C(\epsL,\epsT)}
\newcommand{\frf}{\mathbf{F}_{\mathrm{RF}}}
\newcommand{\fbb}{\mathbf{F}_{\mathrm{BB}}}
\newcommand{\fbbk}[1]{\mathbf{f}_{\mathrm{BB},#1}}
\newcommand{\frfk}[1]{\mathbf{f}_{\mathrm{RF},#1}}
\newcommand{\Bcal}{\mathcal{B}}
\newcommand{\Scal}{\mathcal{S}}

\newtheoremstyle{custom_style} 
  {\topsep} 
  {\topsep} 
  {\normalfont} 
  {} 
  {\bfseries} 
  {} 
  {.2em} 
  {\thmname{#1} {\normalfont\textbf{\thmnumber{#2}}} \thmnote{\normalfont(#3)\textbf{.}}} 

\theoremstyle{custom_style}

\newtheorem{theorem}{\emph{\underline{Theorem}}}
\newtheorem{acknowledgement}[theorem]{Acknowledgement}
\newtheorem{axiom}[theorem]{Axiom}
\newtheorem{case}[theorem]{Case}
\newtheorem{claim}[theorem]{Claim}
\newtheorem{conclusion}[theorem]{Conclusion}
\newtheorem{condition}[theorem]{Condition}
\newtheorem{conjecture}[theorem]{\emph{\underline{Conjecture}}}
\newtheorem{criterion}[theorem]{Criterion}
\newtheorem{definition}{\emph{\underline{Definition}}}
\newtheorem{exercise}[theorem]{Exercise}
\newtheorem{lemma}{\emph{\underline{Lemma}}}
\newtheorem{example}{\emph{\underline{Example}}}
\newtheorem{observation}{\emph{\underline{Observation}}}
\newtheorem{corollary}{\emph{\underline{Corollary}}}
\newtheorem{notation}[theorem]{Notation}
\newtheorem{problem}[theorem]{Problem}
\newtheorem{proposition}{\emph{\underline{Proposition}}}
\newtheorem{solution}[theorem]{Solution}
\newtheorem{method}[theorem]{Method}
\newtheorem{summary}[theorem]{Summary}
\newtheorem{assumption}{\emph{\underline{Assumption}}}
\newtheorem{remark}{\bf \emph{\underline{Remark}}}

\def\qed{$\Box$}
\def\QED{\mbox{\phantom{m}}\nolinebreak\hfill$\,\Box$}
\def\endproof{\hspace*{\fill}~\qed\par\endtrivlist\vskip3pt}

\def\E{\mathsf{E}}
\def\eps{\varepsilon}
\def\Lsp{{\boldsymbol L}}
\def\Bsp{{\boldsymbol B}}
\def\lsp{{\boldsymbol\ell}}
\def\Ltsp{{\Lsp^2}}
\def\Lpsp{{\Lsp^p}}
\def\Linsp{{\Lsp^{\infty}}}
\def\LtR{{\Lsp^2(\Rst)}}
\def\ltZ{{\lsp^2(\Zst)}}
\def\ltsp{{\lsp^2}}
\def\ltZt{{\lsp^2(\Zst^{2})}}
\def\ninN{{n{\in}\Nst}}
\def\oh{{\frac{1}{2}}}
\def\grass{{\cal G}}
\def\ord{{\cal O}}
\def\dist{{d_G}}
\def\conj#1{{\overline#1}}
\def\ntoinf{{n \rightarrow \infty }}
\def\toinf{{\rightarrow \infty }}
\def\tozero{{\rightarrow 0 }}
\def\trace{{\operatorname{trace}}}
\def\ord{{\cal O}}
\def\UU{{\cal U}}
\def\rank{{\operatorname{rank}}}
\def\acos{{\operatorname{acos}}}

\def\SINR{\mathsf{SINR}}
\def\SNR{\mathsf{SNR}}
\def\SIR{\mathsf{SIR}}
\def\tSIR{\widetilde{\mathsf{SIR}}}
\def\Ei{\mathsf{Ei}}
\def\l{\left}
\def\r{\right}
\def\({\left(}
\def\){\right)}
\def\llargebra{\left\{}
\def\rlargebra{\right\}}
\def\lmidbra{\left[}
\def\rmidbra{\right]}
\def\labs{\left|}
\def\rabs{\right|}

\setcounter{page}{1}

\newcommand{\eref}[1]{(\ref{#1})}
\newcommand{\fig}[1]{Fig.\ \ref{#1}}

\def\bydef{:=}
\def\ba{{\mathbf{a}}}
\def\bb{{\mathbf{b}}}
\def\bc{{\mathbf{c}}}
\def\bd{{\mathbf{d}}}
\def\bee{{\mathbf{e}}}
\def\bff{{\mathbf{f}}}
\def\bg{{\mathbf{g}}}
\def\bh{{\mathbf{h}}}
\def\bi{{\mathbf{i}}}
\def\bj{{\mathbf{j}}}
\def\bk{{\mathbf{k}}}
\def\bl{{\mathbf{l}}}
\def\bm{{\mathbf{m}}}
\def\bn{{\mathbf{n}}}
\def\bo{{\mathbf{o}}}
\def\bp{{\mathbf{p}}}
\def\bq{{\mathbf{q}}}
\def\br{{\mathbf{r}}}
\def\bs{{\mathbf{s}}}
\def\bt{{\mathbf{t}}}
\def\bu{{\mathbf{u}}}
\def\bv{{\mathbf{v}}}
\def\bw{{\mathbf{w}}}
\def\bx{{\mathbf{x}}}
\def\by{{\mathbf{y}}}
\def\bz{{\mathbf{z}}}
\def\b0{{\mathbf{0}}}

\def\bA{{\mathbf{A}}}
\def\bB{{\mathbf{B}}}
\def\bC{{\mathbf{C}}}
\def\bD{{\mathbf{D}}}
\def\bE{{\mathbf{E}}}
\def\bF{{\mathbf{F}}}
\def\bG{{\mathbf{G}}}
\def\bH{{\mathbf{H}}}
\def\bI{{\mathbf{I}}}
\def\bJ{{\mathbf{J}}}
\def\bK{{\mathbf{K}}}
\def\bL{{\mathbf{L}}}
\def\bM{{\mathbf{M}}}
\def\bN{{\mathbf{N}}}
\def\bO{{\mathbf{O}}}
\def\bP{{\mathbf{P}}}
\def\bQ{{\mathbf{Q}}}
\def\bR{{\mathbf{R}}}
\def\bS{{\mathbf{S}}}
\def\bT{{\mathbf{T}}}
\def\bU{{\mathbf{U}}}
\def\bV{{\mathbf{V}}}
\def\bW{{\mathbf{W}}}
\def\bX{{\mathbf{X}}}
\def\bY{{\mathbf{Y}}}
\def\bZ{{\mathbf{Z}}}

\def\mA{{\mathbb{A}}}
\def\mB{{\mathbb{B}}}
\def\mC{{\mathbb{C}}}
\def\mD{{\mathbb{D}}}
\def\mE{{\mathbb{E}}}
\def\mF{{\mathbb{F}}}
\def\mG{{\mathbb{G}}}
\def\mH{{\mathbb{H}}}
\def\mI{{\mathbb{I}}}
\def\mJ{{\mathbb{J}}}
\def\mK{{\mathbb{K}}}
\def\mL{{\mathbb{L}}}
\def\mM{{\mathbb{M}}}
\def\mN{{\mathbb{N}}}
\def\mO{{\mathbb{O}}}
\def\mP{{\mathbb{P}}}
\def\mQ{{\mathbb{Q}}}
\def\mR{{\mathbb{R}}}
\def\mS{{\mathbb{S}}}
\def\mT{{\mathbb{T}}}
\def\mU{{\mathbb{U}}}
\def\mV{{\mathbb{V}}}
\def\mW{{\mathbb{W}}}
\def\mX{{\mathbb{X}}}
\def\mY{{\mathbb{Y}}}
\def\mZ{{\mathbb{Z}}}

\def\cA{\mathcal{A}}
\def\cB{\mathcal{B}}
\def\cC{\mathcal{C}}
\def\cD{\mathcal{D}}
\def\cE{\mathcal{E}}
\def\cF{\mathcal{F}}
\def\cG{\mathcal{G}}
\def\cH{\mathcal{H}}
\def\cI{\mathcal{I}}
\def\cJ{\mathcal{J}}
\def\cK{\mathcal{K}}
\def\cL{\mathcal{L}}
\def\cM{\mathcal{M}}
\def\cN{\mathcal{N}}
\def\cO{\mathcal{O}}
\def\cP{\mathcal{P}}
\def\cQ{\mathcal{Q}}
\def\cR{\mathcal{R}}
\def\cS{\mathcal{S}}
\def\cT{\mathcal{T}}
\def\cU{\mathcal{U}}
\def\cV{\mathcal{V}}
\def\cW{\mathcal{W}}
\def\cX{\mathcal{X}}
\def\cY{\mathcal{Y}}
\def\cZ{\mathcal{Z}}
\def\cd{\mathcal{d}}
\def\Mt{M_{t}}
\def\Mr{M_{r}}
\def\O{\Omega_{M_{t}}}
\newcommand{\figref}[1]{{Fig.}~\ref{#1}}
\newcommand{\tabref}[1]{{Table}~\ref{#1}}

\newcommand{\var}{\mathsf{var}}
\newcommand{\fb}{\tx{fb}}
\newcommand{\nf}{\tx{nf}}
\newcommand{\BC}{\tx{(bc)}}
\newcommand{\MAC}{\tx{(mac)}}
\newcommand{\Pout}{p_{\mathsf{out}}}
\newcommand{\nnn}{\nn\\}
\newcommand{\FB}{\tx{FB}}
\newcommand{\TX}{\tx{TX}}
\newcommand{\RX}{\tx{RX}}
\renewcommand{\mod}{\tx{mod}}
\newcommand{\m}[1]{\mathbf{#1}}
\newcommand{\td}[1]{\tilde{#1}}
\newcommand{\sbf}[1]{\scriptsize{\textbf{#1}}}
\newcommand{\stxt}[1]{\scriptsize{\textrm{#1}}}
\newcommand{\suml}[2]{\sum\limits_{#1}^{#2}}
\newcommand{\sumlk}{\sum\limits_{k=0}^{K-1}}
\newcommand{\eqhsp}{\hspace{10 pt}}
\newcommand{\tx}[1]{\texttt{#1}}
\newcommand{\Hz}{\ \tx{Hz}}
\newcommand{\sinc}{\tx{sinc}}
\newcommand{\tr}{\mathrm{tr}}
\newcommand{\MAI}{\tx{MAI}}
\newcommand{\ISI}{\tx{ISI}}
\newcommand{\IBI}{\tx{IBI}}
\newcommand{\CN}{\tx{CN}}
\newcommand{\CP}{\tx{CP}}
\newcommand{\ZP}{\tx{ZP}}
\newcommand{\ZF}{\tx{ZF}}
\newcommand{\SP}{\tx{SP}}
\newcommand{\MMSE}{\tx{MMSE}}
\newcommand{\MINF}{\tx{MINF}}
\newcommand{\RC}{\tx{MP}}
\newcommand{\MBER}{\tx{MBER}}
\newcommand{\MSNR}{\tx{MSNR}}
\newcommand{\MCAP}{\tx{MCAP}}
\newcommand{\vol}{\tx{vol}}
\newcommand{\ah}{\hat{g}}
\newcommand{\tg}{\tilde{g}}
\newcommand{\teta}{\tilde{\eta}}
\newcommand{\heta}{\hat{\eta}}
\newcommand{\uh}{\m{\hat{s}}}
\newcommand{\eh}{\m{\hat{\eta}}}
\newcommand{\hv}{\m{h}}
\newcommand{\hh}{\m{\hat{h}}}
\newcommand{\Po}{P_{\mathrm{out}}}
\newcommand{\Poh}{\hat{P}_{\mathrm{out}}}
\newcommand{\Ph}{\hat{\gamma}}
\newcommand{\mat}[1]{\begin{matrix}#1\end{matrix}}
\newcommand{\ud}{^{\dagger}}
\newcommand{\C}{\mathcal{C}}
\newcommand{\nn}{\nonumber}
\newcommand{\nInf}{U\rightarrow \infty}